%% file: RC.COSN15.tex
\documentclass[prodmode,acmec]{sig-alternate-10pt} 

\usepackage[ruled]{algorithm2e}
\usepackage{amsmath}
\usepackage{amssymb}
\usepackage{booktabs}
\usepackage{multirow}

\SetAlFnt{\small}
\SetAlCapFnt{\small}
\SetAlCapNameFnt{\small}
\SetAlCapHSkip{0pt}
\IncMargin{-\parindent}

\usepackage{float}
\usepackage[caption = false]{subfig}

\newcommand{\ie}{\emph{i.e.},~}
\newcommand{\eg}{\emph{e.g.},~}

\def\F{Figure~}

\newif\ifistr

\istrtrue 

\newtheorem{theorem}{Theorem}
\newtheorem{claim}[theorem]{Claim}
\newtheorem{lemma}[theorem]{Lemma}

\newtheorem{definition}{Definition}

\begin{document}

\markboth{A. Ramachandran and A. Chaintreau}{Knowledge Sharing Economy}

\title{Who Contributes to the Knowledge Sharing Economy?}
\author{
Arthi Ramachandran and Augustin Chaintreau\\
\affaddr{Computer Science Department, Columbia University, New York NY, USA}\\
\texttt{arthir@cs.columbia.edu, augustin@cs.columbia.edu}
}


\maketitle
\begin{abstract}
\input{abstract.tex}
\end{abstract}


\section{Introduction}
\label{sec:intro}
\input{intro.tex}

\section{Who is acquiring new information?}
\label{sec:observation}
\input{observation.tex}

\section{Perishable public goods model}
\label{sec:model}
\input{model.tex}

\section{Equilibrium and Specialization}
\label{sec:analysis}
\input{analysis.tex}

\section{Related Work}
\label{sec:relwor}
\input{relwor.tex}

\section{Conclusion}
\label{sec:conclusion}
\input{conclusion.tex}

\section{Acknowledgements}
We would like to thank Meeyoung Cha for providing access and help on the Twitter Data used for comparison. This material is based upon work supported by the National Science Foundation under grant no. CNS-1254035 and through a Graduate Research Fellowship to Arthi Ramachandran. This research was also funded by Microsoft Research under a Graduate Fellowship.

\bibliographystyle{abbrv}
\bibliography{EC_papers,cosn_papers}

\end{document}


\received{??}{??}{??}


\begin{bottomstuff}

\end{bottomstuff}
\end{document}

%% file: abstract.tex
Information sharing dynamics of social networks rely on a small set of influencers to effectively reach a large audience. Our recent results and observations demonstrate that the shape and identity of this elite, especially those contributing \emph{original} content, is difficult to predict. Information acquisition is often cited as an example of a public good. However, this emerging and powerful theory has yet to provably offer qualitative insights on how specialization of users into active and passive participants occurs.

This paper bridges, for the first time, the theory of public goods and the analysis of diffusion in social media. We introduce a non-linear model of \emph{perishable} public goods, leveraging new observations about sharing of media sources. The primary contribution of this work is to show that \emph{shelf time}, which characterizes the rate at which content get renewed, is a critical factor in audience participation. Our model proves a fundamental \emph{dichotomy} in information diffusion: While short-lived content has simple and predictable diffusion, long-lived content has complex specialization. This occurs even when all information seekers are \emph{ex ante} identical and could be a contributing factor to the difficulty of predicting social network participation and evolution.

%% file: intro.tex

In social network services, such as Twitter and Facebook, the primary commodity produced and exchanged is content and information. While, arguably, much of this process is solely hedonic, these social conversations play an increasingly larger role in today's economy. The revenue of content publishers is now primarily driven by audience originating from online social networks~\cite{Olmstead:2011wz}; brands increasingly channel their products to a targeted audience alongside content exchange~\cite{Liu:2014dw}; new business models aim at integrating with peer connections, sometimes competing with traditional firms in providing accommodation, car ride, or financial services~\cite{Byers:2014vl,Cici:2014gs,An:2014tu}. This is unsurprising since decades of empirical studies, predating any online conversation, have shown how individuals rely on their peers or contacts to acquire information before making a choice. It could be to cast a vote~\cite{PF:1948us}, to keep up to date with new products~\cite{Geissler:2005ce}, or to gather important data in the working place~\cite{Cross:2004wy}. 

Our goal is to understand how individual choices govern how \emph{original} information is produced and acquired in today's social networks. We focus on the domain of identification of news content worth reading, where social connections are massively used. As we are all aware, acquiring original information requires effort and some time investment. Social networks benefit users by making the result of this effort available to more people. Previous studies highlighted that most of the population receives original information from a small set of opinion leaders or influentials. To put it bluntly, only a minority of participants add information to those networks, as opposed to simply listening or passing it on (via, e.g., retweets, likes). Many important open questions remain: In a given network, which users have an incentive to produce more original content? Previous studies have shown that influencers are not easy to differentiate from ordinary users. Can we predict the outcome of such a mechanism, where some users specialize? Are there types of content or networks that favor the formation of an elite?

To answer the above questions, we first conduct an empirical study of original information in news diffusion on social media. We then show how they relate to mathematical analysis of a variant of \emph{public goods}. In contrast with some other goods, most online news are tailored for a particular shelf-life. Our results show that this appears to be one of the primary factors which governs both how activity is distributed, and how multiple types of specialization appear in a dynamic non linear public goods model. We show the following contributions: 
\begin{enumerate}
\item We analyze data from multiple online sources exchanged through Twitter, highlighting the production of original content remains extremely concentrated. Barring institutional accounts, the majority of the original content comes from users with mid-range popularity rather than just the just well known people. In fact, counterintuitively, original content production is skewed towards less active and connected people. We also make the following observations: the size of this active minority in proportion to the audience appears to follow primarily from the shelf-life of the content exchanged. Long term content appears to favor a smaller elite, while short-lived information expands the size of active participants.
(Section~\ref{sec:observation}).
\item Since the availability of news worth reading in a social network exhibits the property of a public good, we propose a simple model that extend public good theory to accommodate investment made by individual players towards a perishable good. We show that it reproduces previous observations and does correlate with the activity we empirically observed.
(Section~\ref{sec:model}).
\item 
This model allows us to answer how specialization occurs in knowledge sharing, even where players are \emph{ex ante} identical. We first prove that a unique Nash Equilibrium exists for sufficiently short-lived content, under a condition related to spectral properties of the social network. 
However, we prove that when content is long-lived, specialization is unavoidable, even with identical players on a symmetric graph. Given the presence of multiple equilibria and sensitivity to initial conditions, predictions are complex. 
(Section~\ref{sec:analysis}).
\end{enumerate}

To the best of our knowledge, our paper is the first example that bridges predictions of the behavior of players in a public good game, with empirical evidence from one of its motivating example: information acquisition. The main novelty of our approach is to model information as a public good with decaying value over time \ie they are perishable goods. As a public good, the utility of information to a user comes from her own contributions as well as those of her neighbors. This new approach allows theory and practice to qualitatively align, in spite of simplistic modeling of user behavior. Perishable public goods create non-linear best response, which makes the analysis more complex, but we hope that this first step can motivate more work in this area. Our work is also, to the best of our knowledge, the first one that analyzes the characteristics and shape of the group of users with an original contribution. This addresses a critical problem as social media are typically described as full of noise and redundancy. Our results may further inform how to promote and reward users for their participation, and mechanisms to design social media which makes user well informed.

%% file: observation.tex
Early studies consider information diffusion as two-step model of information flow, with large cascades originating at institutional sources, followed by a series of connectors. However, more recent results~\cite{Goel:2012io} proved that the vast majority of content is received \emph{directly} from one content originator. Knowledge sharing in social media hence depends on some users to exert effort to acquire \emph{original} information. Original content is obtained externally to the social network, either through search engines, time spent on informal web browsing, or offline conversations. To the best of our knowledge, little is known about the characteristics of the users performing that task, although one expects them to be a minority.

To better understand these dynamics, we analyze two complementary datasets: \textbf{(1) The KAIST dataset} (see \cite{Cha:2010tr} for more details) contains the entire Twitter graph from August 2009 and consists of 8m users and 700m links. Taken over the course of a month, the dataset contains 183m tweets. Of these tweets, we considered only those with urls (37m) since those are the tweets that provide an indication of sharing media on twitter. Further, we filtered the tweets by news domains (\eg \texttt{nytimes.com}). The classification of a domain as news was obtained from the Open Directory Project (\texttt{http://www.dmoz.org/}), a volunteer edited directly of Web links. Each link in the directory is annotated with a top level categories and multiple levels of subcategories. In our analysis, we only took into account the top level category. We kept all the domains with a reasonable number of posts ($>2000$ posts) resulting in 31 domains. We removed domains which did not seem to follow the same definition of news as others (aggregators such as e.g. \texttt{news.google.com} and \texttt{reddit.com}, weather services such as \texttt{weather.gov}, and region specific domains such as \texttt{thehindu.com}). 
While the KAIST dataset provides a holistic view of the media landscape, we complement it with a denser, newer snapshot we collected ourselves: (2) \textbf{The NYT dataset}~(see \cite{May:2014iwa} for more details) contains all the Twitter posts containing a URL from the nytimes.com domain during a full week of December 2011. In parallel, we crawled the follower-followee relationship at the same time in order to construct the URLs that each user received. The final dataset totals 346k unique users receiving a total 22m tweets with URL (including multiplicity). Of these, there are 70k unique links.

\subsection{Imbalanced content creation}

Unsurprisingly, in social media like Twitter, a small fraction of users are responsible for a large part of the activity. To quantify this concentration, we use the Lorenz curve~\cite{Lorenz:1905vb}, or the cumulative share of the top $x$\% of users as a function of $x$, in \F\ref{fig:Lorenz}. Since some domains only cater to niche groups, the fraction $x$ here is measured relative to the domain’s audience size (\ie anyone who received or sent at least one such URL). 
\begin{figure}[h!]
\subfloat{\includegraphics[width = 0.495\textwidth]{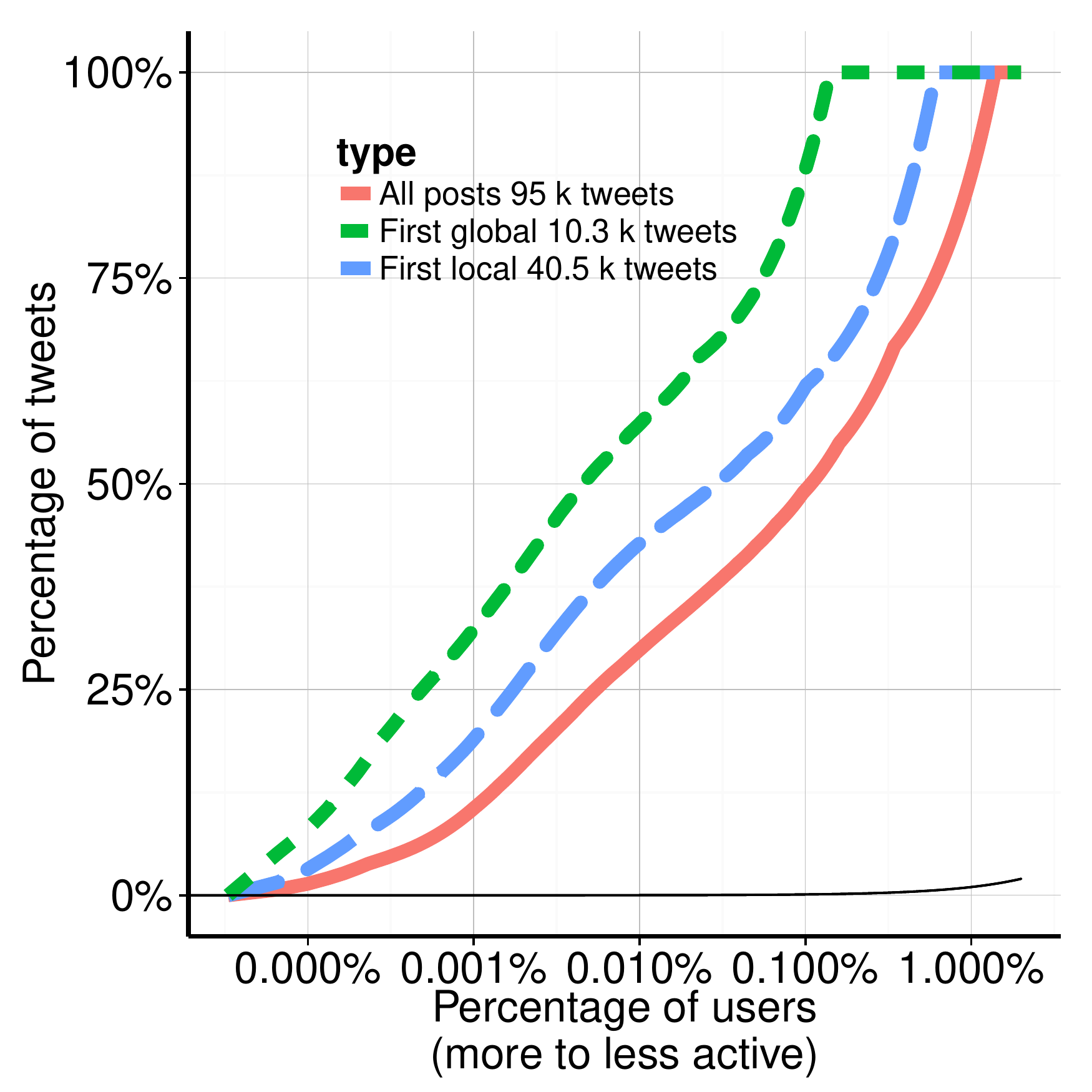}} \\
 \vspace{-15pt}
\caption{Lorenz curve (\ie cumulative share of the top x\% nodes in the audience seen as a function of x): (top) comparing production of tweets and original content for \texttt{cnn.com} from KAIST; (bottom) comparison of Lorenz curve for ``first local’’ tweets in two different domains.}
\label{fig:Lorenz}
\end{figure}
A quick glance at the plot confirms that the size of passive and active audience differ by orders of magnitude (\eg as seen here and in other domains, 99\% do not tweet a single URL. Equivalently, 1\% of the audience produces almost all the new tweets in the network). 

In addition to examining how users post in general (red solid line), we also look at how they acquire original information for the network. We, hence, looked at users who were the first on twitter to post a url link (“global first” represented by the short green dotted line) and users who were the first in their local network, i.e. they did not receive the url from anyone they followed before they sent the url (“local first” represented by the long blue dashed line). Note that in each of these cases, the overall audience remains the same - those who have received the link either directly or indirectly from an originator. Here, in the left figure, 0.1\% of the \texttt{cnn.com} audience produces half of all tweets. But the same number of people produce 60\% of the globally original content and almost 90\% of the locally original content. Perhaps unsurprisingly, while only a small minority of nodes repost articles, it is an even smaller minority that introduces original content in the network. 

\emph{Specialization} is the phenomenon of users taking extreme positions - in our case, some users expend a lot of effort while others are on the other extreme of expending almost no effort. To help quantify this phenomenon, we introduce the 90\%-volume originators measure. This is the fraction of the audience that together produce 90\% of the volume. While we later study how this metric of specialization varies with different content type, we first study the minority of originators in more detail.

\subsection{Characterizing content originators}

It has been shown (see, \eg~\cite{May:2014iwa}) that a user's tweeting activity is strongly correlated with their in- and out-degree. Intuitively, an active online presense is required to gather many followers. Having many followers encourages a return connection by other users. Most Twitter users remain passive in diffusing information, and those promoting original content are a tiny minority. One hypothesis of a simple hierarchy of social media emerges: the content producers responsible for new content creation, the power users and intermediaries who drive the traffic and the passive consumers. As we see here, reality is at odds with this expectation when it comes to production of original content. 

\begin{figure}[!ht]
\subfloat{\includegraphics[width = 3in]{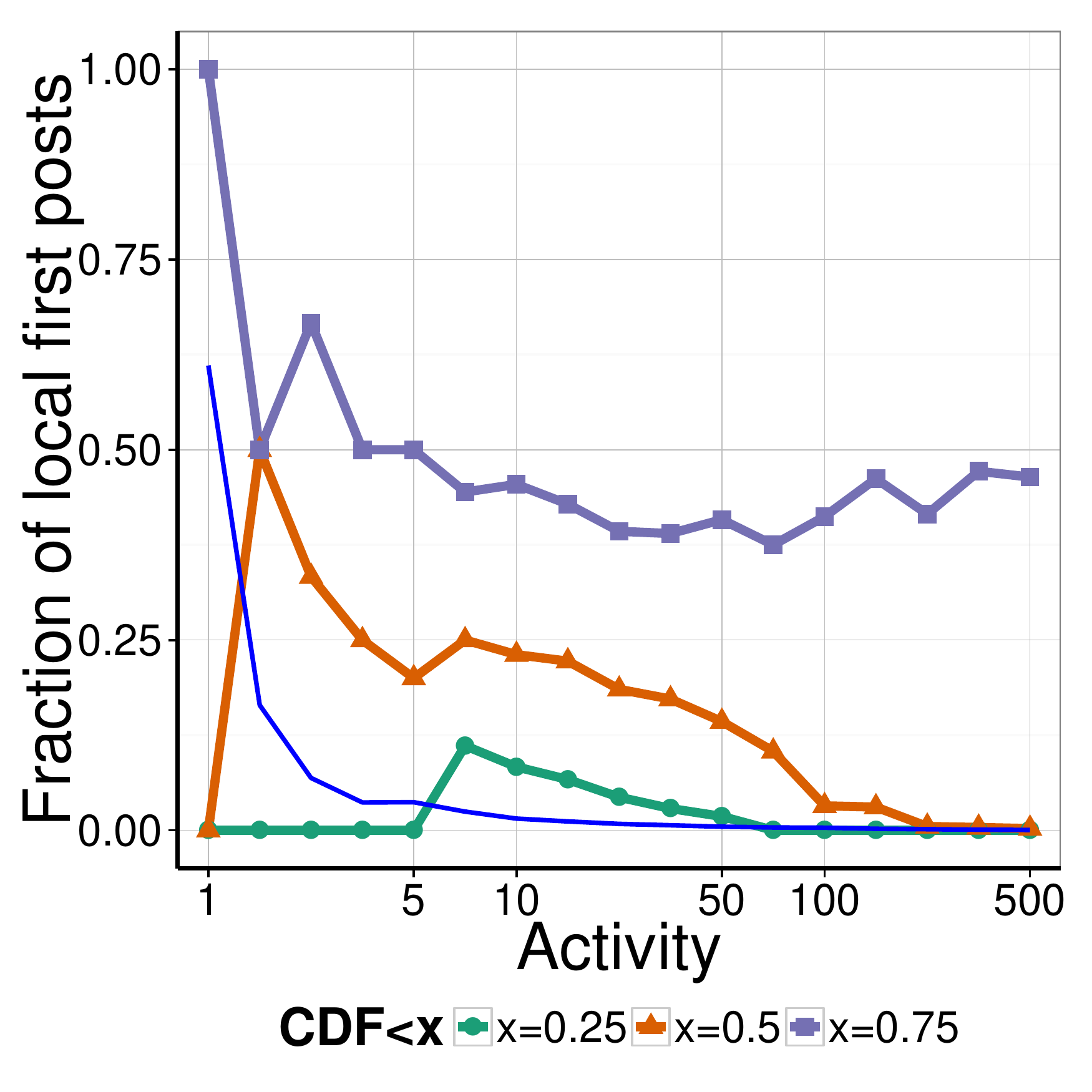}}\\
 \vspace{-10pt}
 \subfloat{\includegraphics[width = 3in]{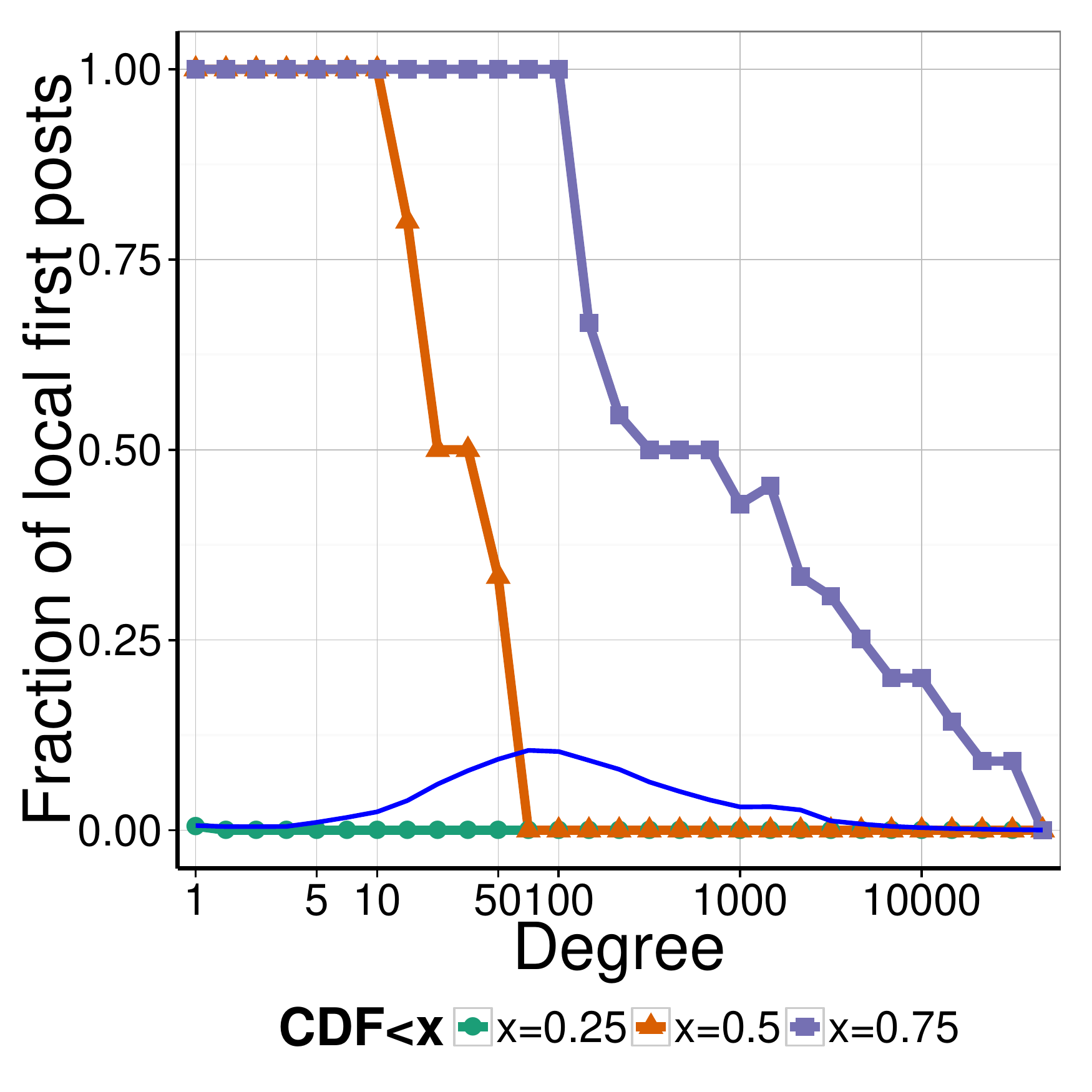}}
 \vspace{-5pt}
 \caption{Fraction of locally original activity, presented as percentiles among users population binned according to (top) activity and (bottom) number of accounts they follow.}
 \vspace{-10pt}
 \label{fig:frac_orig_vs_degree}
\end{figure}

\F\ref{fig:frac_orig_vs_degree} (left) presents, for users binned according to their activity on the x-axis, the distribution of the fraction of local first content they produce with median and various percentiles. To help interpretion, we represent qualitatively with a thin solid line the number of users in each bin, where the first bin contains approximately 129k users. On the right we observe the effects of a few heavy nodes: there are in total 90 users posting more than 400 URLs in a month, who are primarily either institutional accounts or professional journalists and are almost always original. However, those are exceptions: among the active users, originators are generally a minority – typically the 25\% most original – chosen across all activity levels. On the contrary, this trend proves that a URL is most likely to be locally original when it is posted by less active users. Equivalently, if the authors of that tweet post approximately 50 URLs in a month, it is likely to be one she has previously received. Another concurring observation, shown in \F\ref{fig:frac_orig_vs_degree} (right), presents the same distribution where users are binned on the x-axis according to the number of people they follow. The trend here is even more pronounced as users belonging to the less connected half are much more likely to produce original information. 

While, at first, this trend appears relatively surprising, the theory of public goods offers a simple explanation that we leverage later: that the effort exerted by others creates a disincentive for a well connected player to acquire new information. It seems in particular that 50\% of users with larger than average degree rely entirely on the information they receive for their posts.


\subsection{Effect of Time}

Finally, we study the factors quantitatively affecting specialization. To take an example, first, we show in \F\ref{fig:Lorenz2} a comparison between the Lorenz curves for two news media domains: New York Times and The Atlantic. These are different in multiple ways: The New York Times is a daily newspaper with a very large readership while the Atlantic is a monthly magazine with a smaller readership. Within the KAIST dataset, 111k \texttt{nytimes.com} tweets were posted (and an audience of 2.6m users) while 4.7k \texttt{theatlantic.com} tweets were posted (audience of 400k users). Of these tweets only a small fraction are unique links (5917 for \texttt{nytimes.com} vs 891 for \texttt{theatlantic.com}) \cite{Ramach2015}. When comparing lorenz curves, the Atlantic is more specialized than the New York Times with 0.4\% of the audience accounting for 75\% of \texttt{theatlantic.com} tweets while 0.8\% of the audience accounts for 75\% of \texttt{nytimes.com} tweets. This indicates that audiences of different sources specialize in different ways.
\begin{figure}[h!]
\subfloat{\includegraphics[width = 0.495\textwidth]{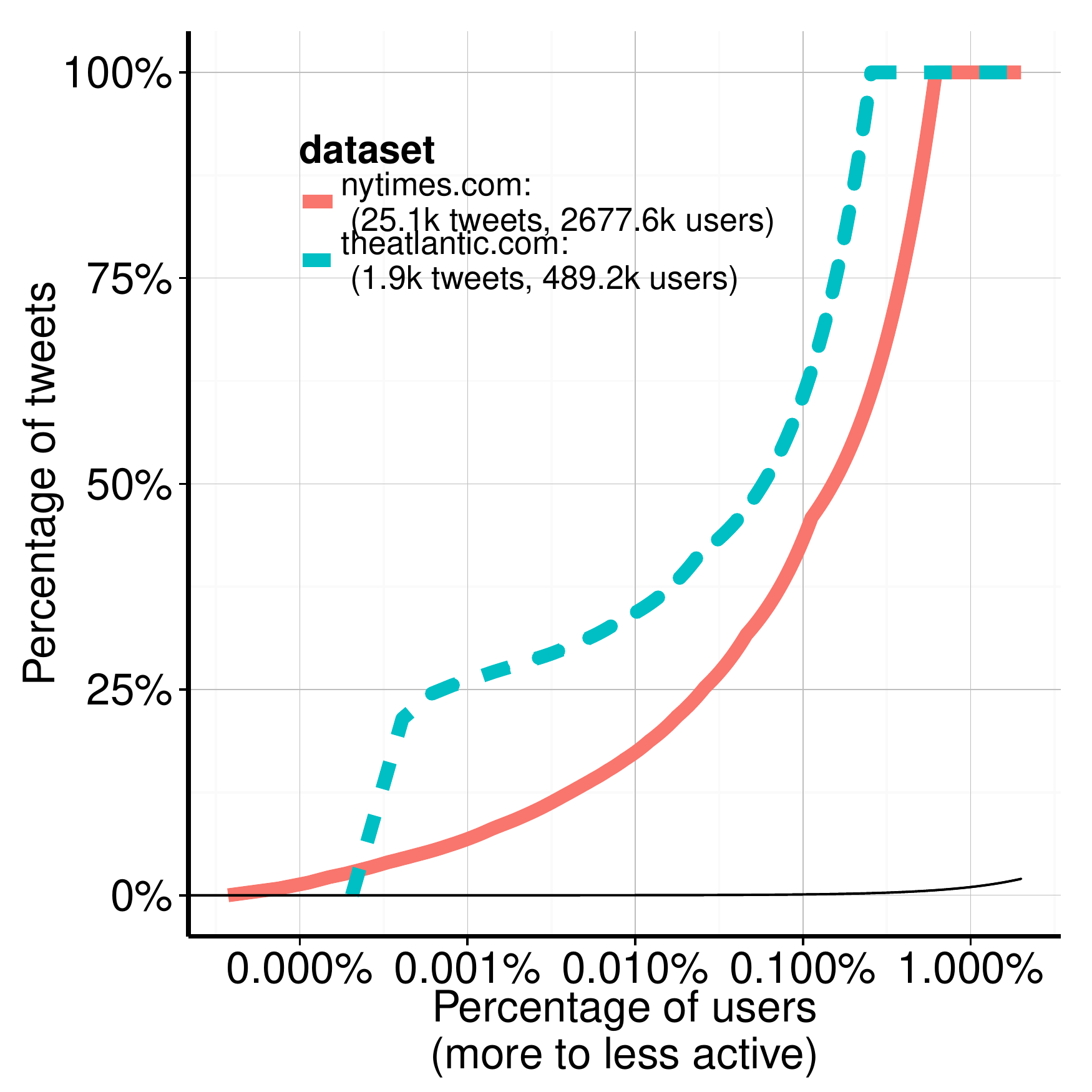}}
 \vspace{-15pt}
\caption{Lorenz curve for ``first local’’ tweets in two different domains.}
\label{fig:Lorenz2}
\end{figure}

\ifistr
\clearpage
\begin{table}[h]
\begin{tabular}{l | lllll}
\multicolumn{1}{c}{domain} & \multicolumn{1}{c}{\begin{tabular}[c]{@{}c@{}}\# unique \\ URLs\end{tabular}}  & \multicolumn{1}{c}{\begin{tabular}[c]{@{}c@{}}\# users \\ (who receive or post)\end{tabular}} & \multicolumn{1}{c}{\begin{tabular}[c]{@{}c@{}}\# users \\ posting\end{tabular}} & \multicolumn{1}{c}{\# posts} & \multicolumn{1}{c}{\begin{tabular}[c]{@{}c@{}}expiration time \\ estimate (min)\end{tabular}} \\
\hline
bbc.co.uk & 19600 & 3252997 & 6248& 113693 & 2.20\\
businessweek.com & 777& 622615 & 927 & 9405& 55.60\\
cnn.com & 10255 & 3026569& 4458& 94965 & 4.21\\
csmonitor.com& 337& 460161 & 492 & 3561& 128.19 \\
economist.com& 232& 802242 & 922 & 3714& 186.21 \\
forbes.com& 1934& 921576 & 1198& 12375 & 22.34 \\
foxnews.com & 1529& 510935 & 2845& 19383 & 28.25\\
ft.com& 6750& 1373373& 2647& 28497 & 6.4\\
guardian.co.uk & 5612& 2106241& 2294& 45911 & 7.70\\
huffingtonpost.com & 3742& 1443562& 2492& 36974 & 11.54\\
mirror.co.uk & 1306& 638255 & 708 & 4863& 33.08\\
news.yahoo.com & 7684& 1467238& 5227& 65734 & 5.62 \\
newsweek.com & 517& 783171 & 679 & 3465& 83.56 \\
newyorker.com& 299& 754866 & 656 & 2444& 144.48 \\
npr.org & 447& 1220573& 1066& 12100 & 96.64  \\
nytimes.com & 5917& 2677563& 4085& 111674& 7.30 \\
online.wsj.com & 5077& 1394111& 2075& 37581 & 8.51 \\
reuters.com & 16634 & 1435299& 2621& 61955 & 2.60 \\
salon.com & 803& 1082391& 745 & 4501& 53.80 \\
slate.com & 518& 676407 & 897 & 3097& 83.40 \\
theatlantic.com& 5917& 489222 & 804 & 4670& 7.30\\
theonion.com & 795& 1427288& 1238& 11969 & 54.34 \\
time.com& 2293& 530981 & 4370& 14299 & 18.84 \\
usatoday.com & 3281& 1141070& 1570& 20912 & 13.17 \\
usnews.com& 1089& 580657 & 373 & 4222& 39.67 \\
vanityfair.com & 162& 598879 & 743 & 2261& 266.67 \\
washingtonpost.com & 2886& 1755915& 2051& 35554 & 14.97 \\
wired.com & 1751& 1325465  & 2307    & 17640   & 24.67
\end{tabular}
\caption{Expiration times of different news sources}
\label{tab:domain_props}
\end{table}
\clearpage
\fi

\begin{figure*}[!ht]
\subfloat{\includegraphics[width=0.9\textwidth]{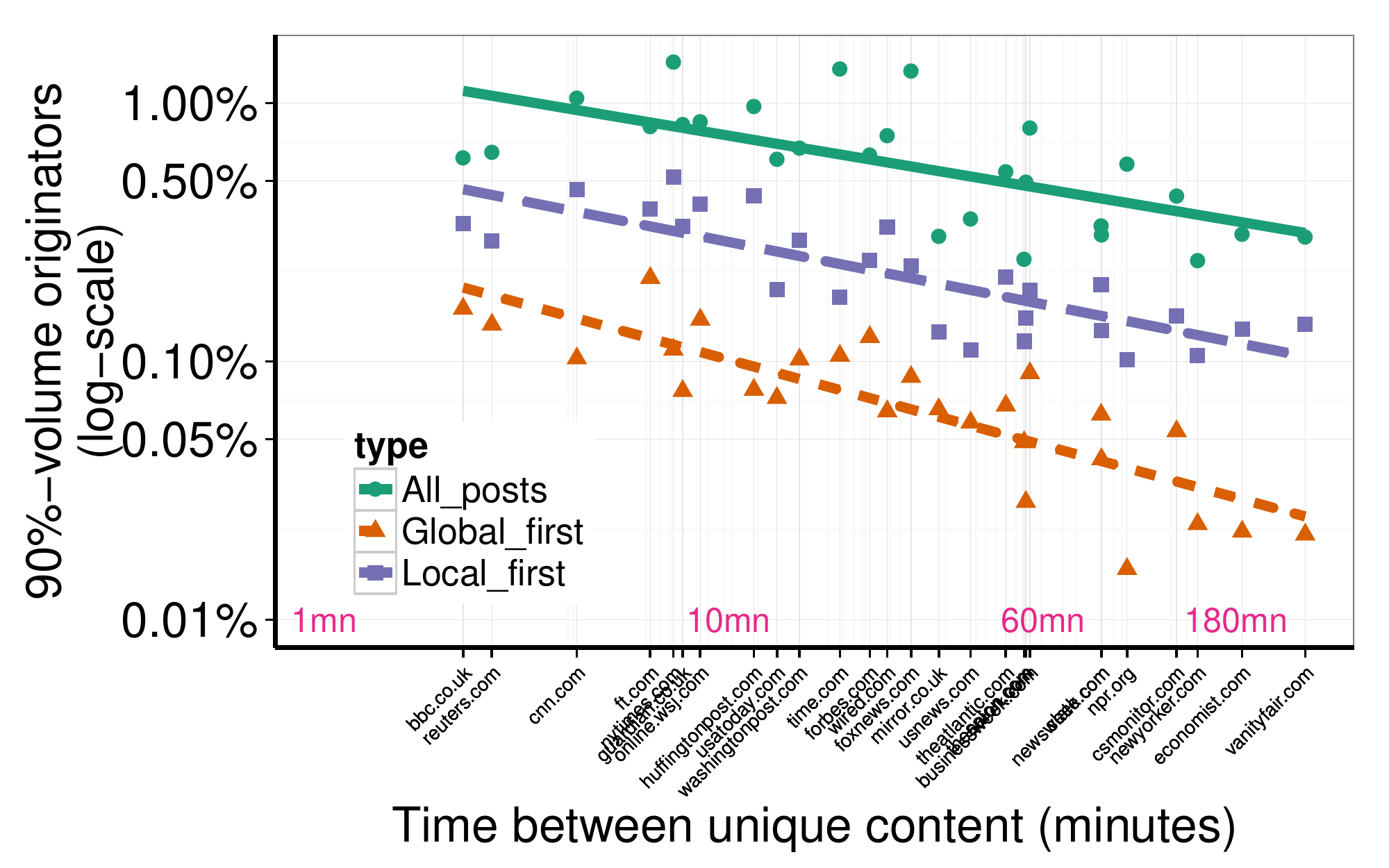}}
\vspace{-10pt}
\caption{Concentration of sharing compared to the shelf-life for each media source. Each point is the fraction of the audience responsible for 90\% of the tweet volume of the media source.}
\label{fig:kaist_shelflife}
\end{figure*}
Our main observation is as follows: the degree of specialization is related to the temporal dynamics of the content, with remarkable regularity. In the same time period, more new content is introduced by \texttt{nytimes.com}, indicating that the content becomes stale quicker than for \texttt{atlantic.com}. This is consistent with \texttt{nytimes.com} being a daily news source. For every media, we measure its average \emph{shelf life} by using the number of unique URLs produced over a month. We define the shelf life of an article to be the amount of time for which it is relevant \ie it continues to be shared among users. This captures the fact that, since all media compete for attention within the same online network, one producing ten times more content expects the content to be renewed ten times faster. \F\ref{fig:kaist_shelflife} shows the 90\%-volume originator (\ie the percentage of the audience producing 90\% of tweet volumes) for 31 media sources. There is a fairly large range of shelf life from approximately 2 minutes to over 2 hours. However, we consistently observe that domains with long shelf times tend involve a smaller fraction of the population to produce most of the content. Note that the x- and y-axis are in logscale. This temporal dynamics affects all tweets and original content similarly. After renormalization, this seems not be affected much by audience size, although we did observe the smaller effect of the fraction of active users grows slowly with the audience.

We also examined the effect of different measures of shelf lives in \F\ref{fig:kaist_shelflife_perarticle}. We calculate the diffusion life as the length of time that the article is shared (time of last post - time of first post). The y-axis is a measure of concentration, fraction of locally first posts of the total number of people receiving the article. We normalized by the number of users posting the article, in order to better account for larger cascades. Other measures of concentration, such as the fraction of first local posts by the total number of posts of an article, also exhibit similar trends, albeit in a more muted fashion. We continue to see the trend of articles with longer shelf lives tend to be more concentrated in sharing.

\begin{figure}[!ht]
\centering \subfloat{\includegraphics[width=0.4\textwidth]{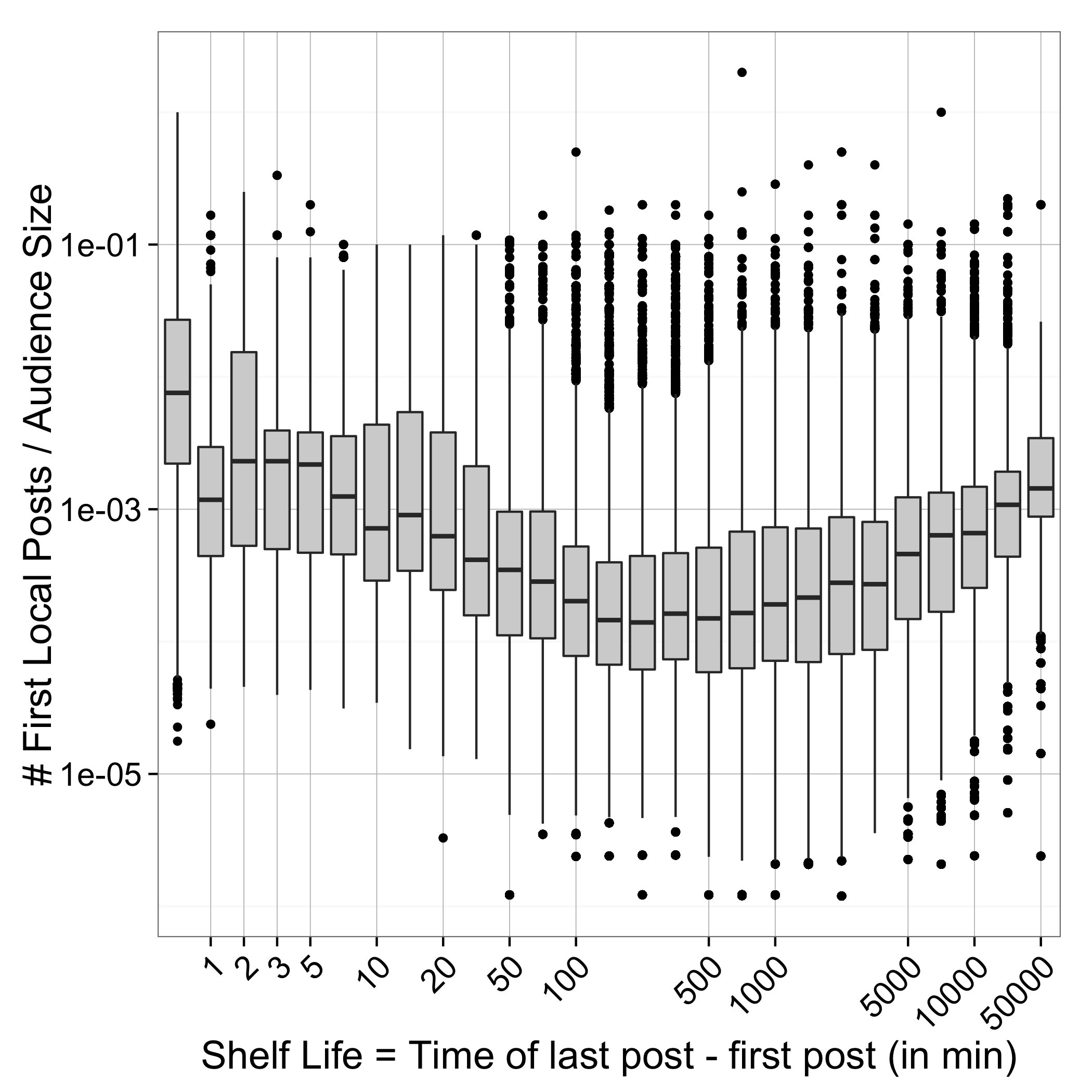}}
\vspace{-10pt}
\caption{Concentration of sharing compared to the diffusion-life for each article. Each article's diffusionlife is the total active time (in minutes) of the article.}
 \vspace{-15pt}
\label{fig:kaist_shelflife_perarticle}
\end{figure}


In summary, we have made several observations. (1) The presence of specialization where a small number of individuals are responsible for most of the original content produced on Twitter. (2) These individuals who produce most of the original content are not, as expected at first glance, the most well connected or the highest degree nodes. Rather, they are average-degree nodes in the network. (3) There is a correlation between the shelf life of an article, the time for which it is relevant, and the degree of specialization. To the best of our knowledge, there does not exist a previous model with reproduces these characteristics. In the following section, we present an idealized model which retains the flavor of the problem of information search.

%% file: model.tex

While information diffusion on social media is complex and topic dependent, our goal in this section is to provide a simple model with which previous observations of information acquisition can be predicted. We leverage the economic theory of public goods -- goods that are non-rivalrous where use by one individual does not reduce availability to others. In fact, in many public goods models, the ownership of the good by on individual has an impact on the utility of his neighbors. Further, we consider news as a perishable good, i.e. a good that needs to be used within a short period of time and bought again (such as milk or produce). While news does not spoil in the same sense as produce does, the value of news does decrease with time due to updated information and later events occurring. In both cases, since the product is short-lived and the demand is persistent, there is a time dynamic to renew it.


\subsection{A Public Good Approach to Original Content Production}

As content online is vast and not easy to navigate, we assume that player $i$ seeks knowledge at a given rate. This results in content being discovered by her at random times with an intensity $y_i$, forming a Poisson process of discovery times. The effort of that user to individually achieve a discovery rate $y_i$ has a convex cost $c(y_i)$. This captures the fact that as more effort is exerted, or time is invested, worthwhile information becomes rare and harder to find. The utility of information is represented as being in an informed state. In this state, a user has an additional unit of return compared to being uninformed. Upon a discovery, a user remains in the informed state for a time $\tau$ equal to the shelf time of this item. We assume $\tau$ is a constant. 

There is a social component to the interaction: users make the results of their work available to neighbors in a social network graph. We denote the adjacency matrix of the social network as $G=(V,E)$ and it can either be undirected (\eg Facebook) or directed (\eg Twitter). Without loss of generality, we assume that the effort of a user only affects its direct neighbors. The general case simply requires redefining neighboring relations to include future descendents.

Let us denote $y_{-i}$ = $\sum_{j\in N(i)}{y_j}$ as the rate of content discovery that a user $i$ in the network receives at no cost from her neighbors. Then, including her own effort cost $c(y_i)$, the average utility received per unit of time can be written as:
$$U(y_i, y_{-i}) = 1 - e^{-\tau(y_i+y_{-i})} - c(y_i) 
	\;.$$
At time $t=T$, the probability to have received one content item within $]T-\tau;T]$ is the probability that a Poisson process of rate $(y_i+y_{-i})$ creates no point in that interval.

Note here, that discovering multiple content simultaneously creates no additional benefit to the user since the user is already in the informed state. Note also that having content items of various shelf-lives would result in the same dynamics as long as those durations are chosen independently of the discovery process. Finally, while most of the properties of the model we show generalizes to general convex cost, we are primarily interested in polynomial cost $(c:y_i \mapsto\frac{\theta}{\alpha+1} y_i^{\alpha+1}), \alpha>0$. 
We can think of $\theta$ as the reference time period. A reward of $1$ is equivalent to the effort spent to produce content once every $\theta$ time. In this work, we assume, in general, that the cost is normalized such that $\theta=1$hr. This means that the reward exactly compensates for the search effort incurred to produce original content every hour. More general models, especially ones with heterogeneous costs and a matrix of benefits transfer between users, are likely to perfect realism of this model, but we leave them for future work.

\subsection{Best Response}

We first analyze a single individual response of a player to her neighbors' efforts. Even with non-linear dynamics is non-linear, we can represent this best response action in a simple closed form.
\begin{theorem}
For a node, $i$, of $G=(V,E)$, the best response to $i$'s neighbors' efforts, $y_{-i}$, is given by \\ $\phi(y_{-i},\tau) = \frac{\alpha}{\tau}W(\frac{\tau^{\frac{\alpha+1}{\alpha}}}{\alpha} e^{-\frac{\tau y_{-i}}{\alpha}})$, where $W$ is the Lambert function defined on $[0;\infty[$ as the inverse of the function $x\mapsto x\exp(x)$.
\end{theorem}
\begin{proof}
For an individual, $i$, their best response to their neighbors efforts occurs when $i$'s utility is maximized w.r.t. the amount of effort $i$ invests, $y_i$.
\[
\max_{y_i}{U(y_i, y_{-i})} \;\textrm{s.t.}\; y_i \ge 0 
\;\textrm{, \ie ,}\;
\frac{\partial U(y_i, y_{-i})}{\partial y_i} =  0\;.
\]
This yields $\tau e^{-\tau(y_i+y_{-i})} - y_i^\alpha = 0$. \\ Hence 
$\tau y_i = \alpha W(\frac{\tau^\frac{\alpha+1}{\alpha}}{\alpha} e^{-\frac{\tau y_{-i}}{\alpha}}) $
where $W$ denotes the Lambert function, which proves the result.
\end{proof}

The Lambert function $W$ (\F\ref{fig:lambert}) is a positive increasing function, that is asymptotically equivalent to the identity near $0$ and comes within a negligible distance of the function $x\mapsto \ln(x)-\ln\ln(x)$ as $x$ becomes large. The last two decades has found numerous applications of this function to differential equation, combinatorics, theoretical physics and others. Its computation, both through formal calculus and numerical approximation can be done fast.

Our closed form implies the bound for any $y : 
0=\lim_{x\rightarrow \infty}{\phi(x)}
\leq 
\phi(y)
\leq 
\phi(0)=\frac{\alpha}{\tau}W(\tau^\frac{\alpha+1}{\alpha}) 
\;.$

\begin{figure}[!ht]
\vspace{-10pt}
\subfloat{\includegraphics[width=.245\textwidth]{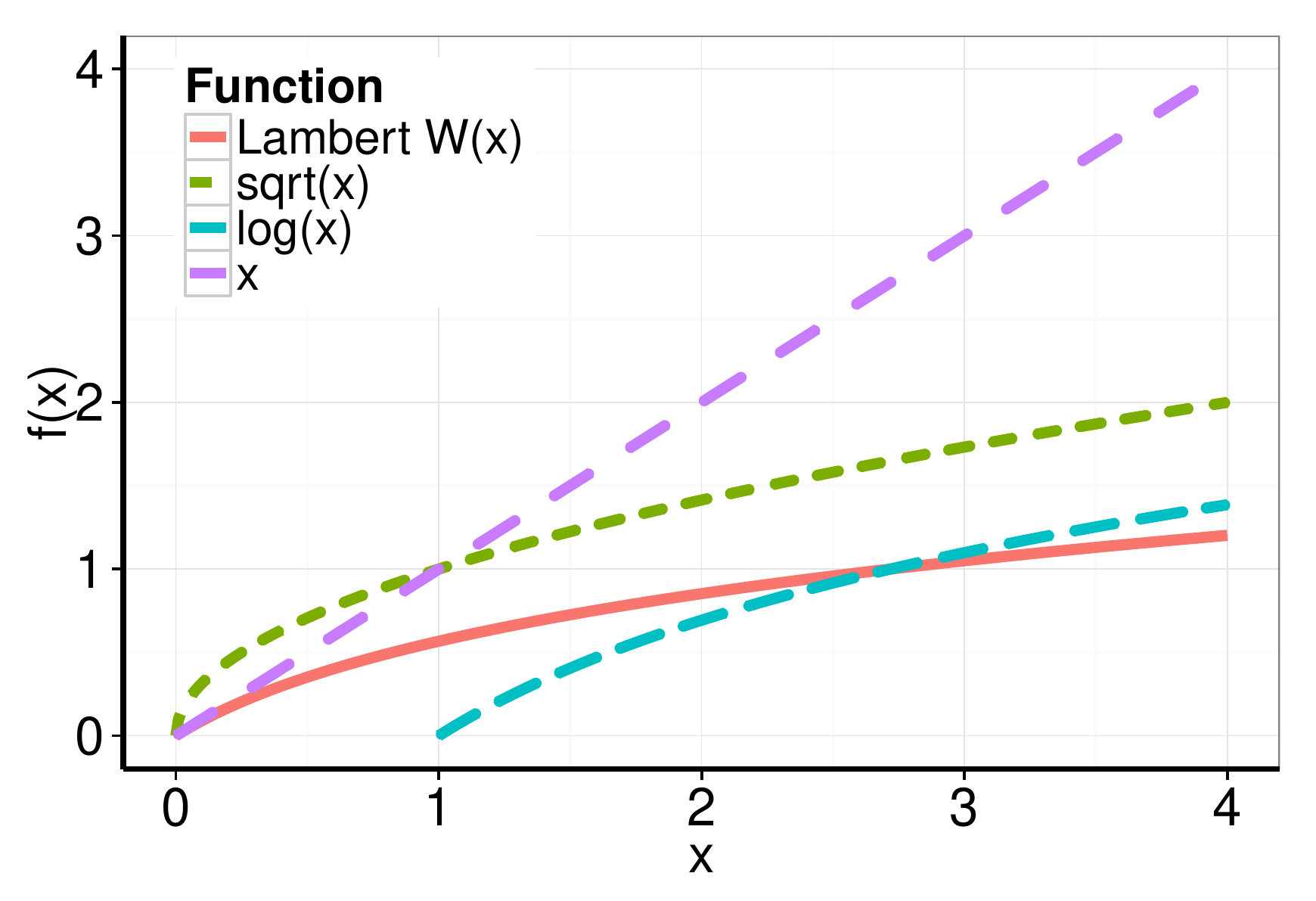}}
\subfloat{\includegraphics[width=.245\textwidth]{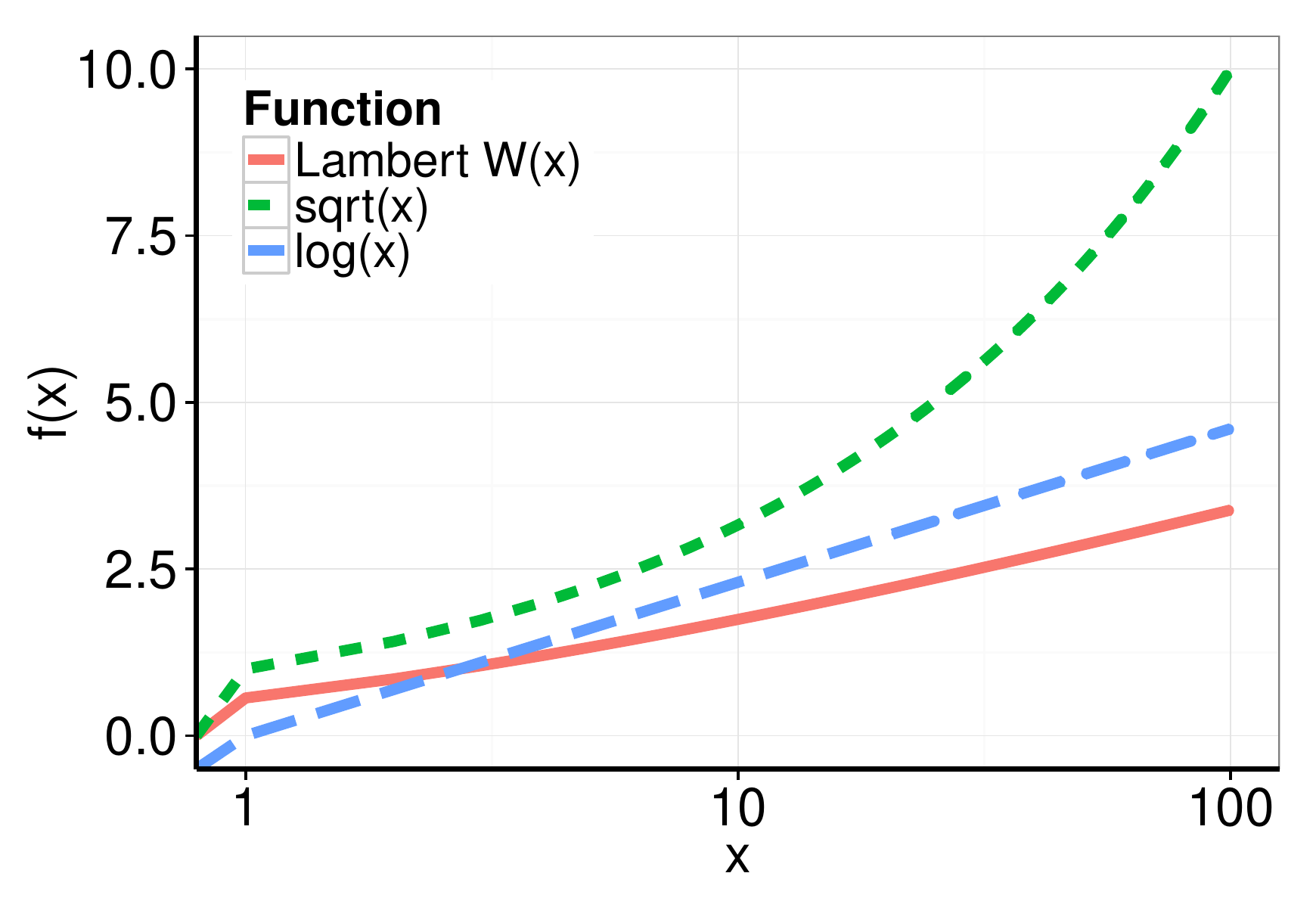}}
\vspace{-15pt}
\caption{Comparison of the Lambert function ($W(z)$ where $z = W(z)e^{W(z)}$) to the common function of $x$, $\log{(x)}$, and $\sqrt{x}$ in the range (left) [0,4] and (right) [1,100]. }
\label{fig:lambert}
\end{figure}

\subsection{Nash Equilibrium}
We initially focus on analyzing the Nash equilibrium in symmetric graphs. 
\begin{definition} 
A graph $G$ is symmetric if, given any two pairs of edges $(u_1, v_1)$ and $(u_2, v_2)$ of $G$, there is an automorphism $f : V (G) \rightarrow V (G)$ such that $f(u_1) = u_2$ and $f(v_1) = v_2$.
\end{definition}
In a symmetric graph, in a unique Nash Equilibrium, all nodes exert the same amount of effort. Observe that if this were not the case, a transformation of the graph results in another equilibrium. 
\begin{lemma} 
For a $D$-regular graph, a symmetric Nash Equilibrium always exists and is given by $$y_i=\frac{\alpha}{\tau (1+D)}W(\tau^\frac{\alpha+1}{\alpha} \frac{(1+D)}{\alpha}), \forall i.$$
\label{lem:symm_ne}
\end{lemma}
\ifistr
\begin{proof}
In a symmetric equilibrium, $y_i=y,\forall i \in G $. Also, for a node $i$, $y_{-i}=Dy$.
\begin{align*}
\text{At equilibrium} \quad y &= \frac{\alpha}{\tau} W(\frac{\tau^\frac{\alpha+1}{\alpha}}{\alpha} e^{-\frac{\tau Dy}{\alpha}})  \\
\frac{\tau}{\alpha}y e^{\frac{\tau}{\alpha}y } &=\frac{\tau^{\frac{\alpha+1}{\alpha}}}{\alpha} e^{-\frac{\tau Dy}{\alpha}}\\\
y \frac{\tau}{\alpha}(1 +D) &= W(\frac{\tau^{\frac{\alpha+1}{\alpha}}}{\alpha}(1 +D))\\\
y &= \frac{\alpha}{\tau(1+D)}W(\frac{\tau^{\frac{\alpha+1}{\alpha}}}{\alpha}(1 +D))
\end{align*}
\end{proof}
\fi
\vspace{-5mm}
The case of symmetric graphs is interesting because, as we show in Section \ref{sec:uniq_ne_cond}, this symmetric equilibrium need not always be a unique or stable equilibrium.

\subsection{Model Validation}
Real world graphs are, of course, more complex than the above symmetric graph models. We validate our model on a subset of the NYT graph (a random sample of 10\% of the edges). We use an iterative update method (described in the long version of this paper \cite{Ramach2015}) to find the Nash equilibrium numerically. In these simulations, we used a range of shelf-life times ranging from short ($\tau =1$) to long ($\tau=1000$).

Matching our observations from the KAIST dataset, users with larger degree have less ``information seeking activity''. This is reflected in a smaller amount of effort spent in the Nash Equilibrium. Figure \ref{fig:nytne} (left) shows the correlation of the Nash Equilibrium effort with out-degree of a node ($\tau=0.5$ on a sample of 0.1\% of the NYT graph). Here, we see a very strong relationship between the degree and the amount of effort expended in the Nash Equilibrium. Thus, our model yields predictive power for relation of connection and investment in information search

\begin{figure}[ht!]
\vspace{-10pt}
\subfloat{\includegraphics[width=0.245\textwidth]{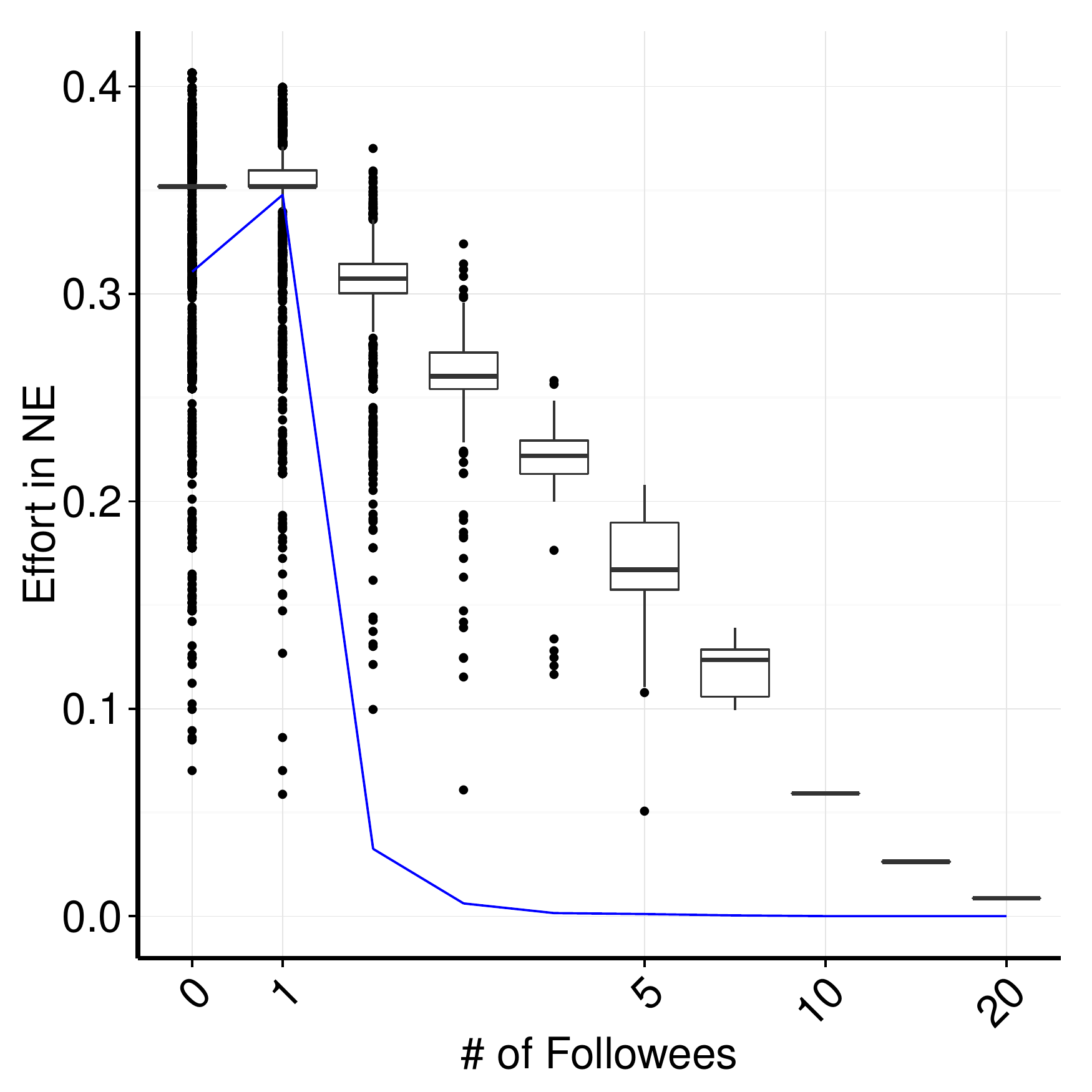}}
\subfloat{\includegraphics[width=0.245\textwidth]{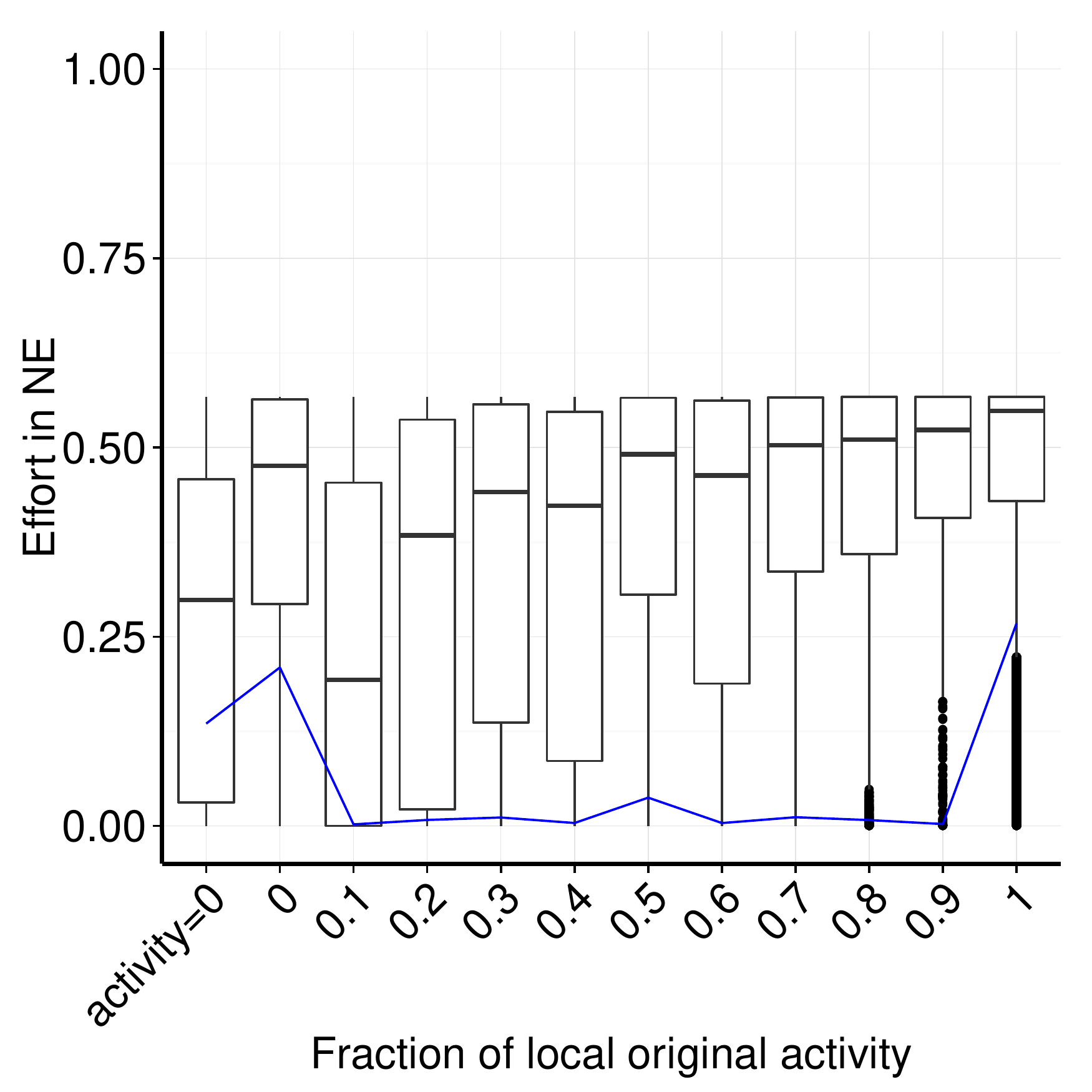}}\\
\vspace{-20pt}
\caption{The Nash Equilibrium (as a function of (left) node degree and (right) fraction of first local activity) in a sample of the NYT graph}
\label{fig:nytne}
\end{figure}

We then observe that the elite in the modeled equilibrium share similar structure to those observed empirically (Figure~\ref{fig:nytshelflife}). A small subset of individuals are responsible for a large fraction of the effort spent -- mimicing the behavior of individuals with original content. 

\begin{figure}[ht!]
\subfloat{\includegraphics[width=0.45\textwidth]{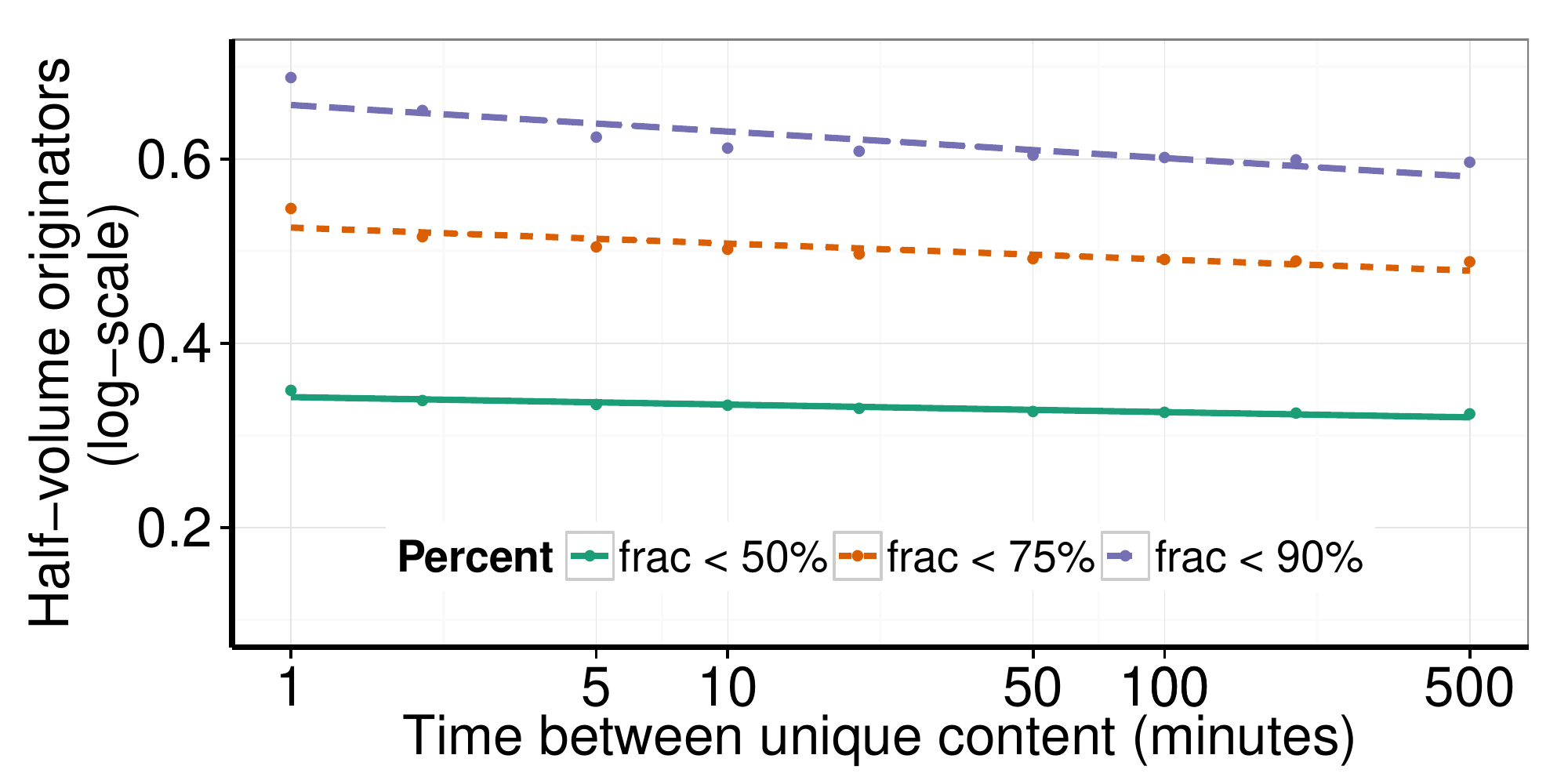}}
\vspace{-10pt}
\caption{Proportion of population responsible for 50\%, 75\% and 90\% of the effort in the Nash Equilibrium in sample of NYT graph.}
\label{fig:nytshelflife}
\end{figure}

Lastly, we examine how the effort in the Nash equilibrium of our model correlates to the fraction of local original activity vs total activity observed in the NYTimes dataset (Figure~\ref{fig:nytne} right). Ideally, we would expect to see perfect correlation since the effort in our model captures exactly this, the effort you spent to bring new content to your neighbors. We see that individuals who in the real world had no effort (the left most group) expend low effort in the Nash Equilibrium. Those who posted at least one article expended more effort and the amount of effort steadily rises.

%% file: analysis.tex

\subsection{Conditions for a Unique Nash Equilibrium}
\label{sec:uniq_ne_cond}
Different classes of goods exhibit different types of behavior. In economic theory, one of these classifications are that of a \emph{normal good} is a good for which demand increases with increased wealth. Mathematically, if $\gamma:\mathbb{R}_{\ge 0} \rightarrow \mathbb{R}_{\ge 0}$ is a differentiable function representing the income elasticity of demand (the responsiveness of the demand to a change in the income), then the good is normal iff the derivative satisfies $0<\gamma'<1$. 
A \emph{network normal good} carries that idea to a networked case where there is a income elasticity of demand function for each player $i$ in the network. The consumption $\gamma_i$ is defined in terms of the wealth of $i$ (set externally), $w_i$, and $i$'s ``social income'', the income from neighbors of $i$, $y_{-i}$. A network normal good satifies the condition: $1 + \frac{1}{\lambda_{\min}}<\gamma'_i(w_i+y_{-i})<1$ \cite{Allouch:2012do}. We can also express these conditions in terms of the best response $\phi(y_{-i}) =\gamma_i(w_i+y_{-i}) -w_i$ as follows.

\vspace{2mm}
\noindent {\bf Fact.} In the above notation, a good is network normal iff for every player $i$,  $\frac{1}{\lambda_{\min}}<\phi'(y_{-i})<0$.

In our model, there can exist multiple equilibria for the effort that individuals expend. Using network normality conditions, we now give a condition involving the expiration time parameter, $\tau$ under which the Nash equilibrium for the system will be unique. 
\ifistr
\begin{lemma}
$\frac{\partial \phi}{\partial y} = -\frac{W(\tau^2 e^{-\tau y})}{1+W(\tau^2 e^{-\tau y})}$
\label{lem:lambert_deriv}
\end{lemma}
\begin{proof}
\begin{align*}
\frac{\partial \phi}{\partial y} &= \frac{\partial \big(\frac{1}{\tau}W(\tau^2e^{-\tau y})\big)}{\partial y}\\
&=\frac{1}{\tau}\cdot \tau^2\cdot (-\tau)\cdot e^{-\tau y}\cdot W'(\tau^2 e^{-\tau y})\\
&=-\tau^2 e^{-\tau y}\cdot \frac{W(\tau^2 e^{-\tau y})}{\tau^2 e^{-\tau y}(1+W(\tau^2 e^{-\tau y}))}\\
&=-\frac{W(\tau^2 e^{-\tau y})}{1+W(\tau^2 e^{-\tau y})}
\end{align*}
\end{proof}
\fi

\begin{theorem} \ifistr [Short-Lived Content Exhibits Less Specialization] \fi
Let $\lambda_{\min}$ be the minimum eigenvalue of the adjacency matrix of the network, $G=(V,E)$, and let $\tau$ be the expiration time parameter of the system. Then, a unique Nash Equilibrium exists if 
$$\tau < \hat{\tau} =^{\text{def}} \big(\frac{\alpha}{-\lambda_{\min} -1}\big)^\frac{\alpha}{\alpha+1}e^{\frac{\alpha}{(\alpha+1)(-\lambda_{\min} -1)}}\;.$$
\end{theorem}
\begin{proof}
We will prove the theorem by using the previously established connection between network normality of the system and the existence of a unique Nash equilibrium \cite{Allouch:2012do, Bramoulle:2007va, Bramoulle:2010uj, Bergstrom:1986ej}. Hence we only need to show that the network normal conditions hold under the assumptions of the theorem.

We will show that the condition holds for every player, $i$. For ease of notation, let $\phi=\phi_i$ and $x=y_{-i}$.  

 Observe that since $W$ is an increasing function, we have $\phi'(x)$ is a non-decreasing function. Hence the derivative only takes values in $[\phi'(0), \lim_{x \rightarrow \infty}\phi'(x)]=[-\frac{\tau^\frac{\alpha+1}{\alpha}}{\alpha}W'(\frac{\tau^\frac{\alpha+1}{\alpha}}{\alpha}), 0]$. Now, the network normality condition simplifies to verifying $$ \frac{1}{\lambda_{\min}(G)} < -\frac{\tau^\frac{\alpha+1}{\alpha}}{\alpha}W'(\frac{\tau^\frac{\alpha+1}{\alpha}}{\alpha}) < 0.$$
Simplifying the first inequality, we get:
\begin{align*}
\tau&< \big(\frac{\alpha}{-\lambda_{\min} -1}\big)^\frac{\alpha}{\alpha+1}e^{\frac{\alpha}{(\alpha+1)(-\lambda_{\min} -1)}}=\hat{\tau}
\end{align*}
Thus, the network normality conditions holds and a unique Nash equilibrium exists for any $\tau<\hat{\tau}$.
\end{proof}

\ifistr
Here, we show the conditions necessary for a unique Nash Equilibrium to exist for various graph families. Let $\lambda_{\min}$ be the minimum eigenvalue of the adjacency matrix of the network, $G=(V,E)$, and let $\tau$ be the expiration time parameter of the system. Then, a unique Nash Equilibrium exists if $$\tau < \hat{\tau} =^{\text{def}} \big(\frac{1}{-\lambda_{\min} -1}\big)^\frac{\alpha}{\alpha+1}e^{\frac{\alpha}{(\alpha+1)(-\lambda_{\min} -1)}}\;.$$ 

\begin{lemma}
A complete graph always has a unique Nash equilibrium
\end{lemma}
\begin{proof}
In a complete graph, $\lambda_{\min}=-1$. Thus, for any value of $\tau$, there exists a unique Nash equilibrium.
\end{proof}

\begin{lemma}
In a star graph with $n-1$ leaf nodes, a unique Nash equilibrium for $\tau < \hat{\tau}=(\frac{1}{\sqrt{n-1}-1})^{\frac{1}{2}}e^{\frac{1}{2(\sqrt{n-1}-1)}}$,
\end{lemma}
\begin{proof}
In a star graph of size $n$, $\lambda_{\min} =-\sqrt{n}$ (\cite{Brouwer:2012wz}).
\begin{align*}
W(\tau^2) &<\frac{1}{\sqrt{n}-1}\\
\tau^2 &<\frac{1}{\sqrt{n}-1}e^{\frac{1}{\sqrt{n}-1}}
\end{align*}
\end{proof}

\begin{lemma}
An even cycle graph of size $n$ has a unique Nash equilibrium for $\tau<\hat{\tau}=\sqrt{e}$.
\end{lemma}
\begin{proof}
\begin{align*}
\lambda_{\min}&=-2 \quad (\cite{Brouwer:2012wz}).\\
W(\tau^2) &< 1\\
\tau &= \sqrt{e}
\end{align*}
\end{proof}

\begin{lemma}
An odd cycle graph of size $n$ has a unique Nash equilibrium for $\tau<\hat{\tau}=\frac{n}{(n^2 - \pi^2)^{\frac{1}{2}}}e^{\frac{n^2}{2(n^2 - \pi^2)}}$.
\end{lemma}
\begin{proof}
\begin{align*}
\lambda&=2 \cos{\frac{2\pi j}{n}} \quad (\cite{Brouwer:2012wz}) &(j={0,1,...,n-1})\\
\lambda_{\min}&=2\cdot \cos{(\pi-\frac{\pi}{n})} &(j=\frac{n-1}{2})\\
\lambda_{\min}&=-2\cdot \cos{\frac{\pi}{n}} &(\because \cos(\pi-\theta)=-\cos(\theta))\\
\lambda_{\min}&=-2\cdot \sum_{n=0}^{\infty}{(-1)^n\frac{\frac{\pi}{n}^{2n}}{(2n)!}} &(\text{Taylor expansion})\\
\lambda_{\min}&\approx-2\cdot (1-\frac{\pi^2}{2!n^2}+\frac{\pi^4}{4!n^4})\\
\lambda_{\min}&\approx-2 + \frac{\pi^2}{n^2}
\end{align*}
Substituting the value for $\lambda_{\min}$,
\begin{align*}
W(\tau^2) &<\frac{1}{1 - \frac{\pi^2}{n^2}}\\
\tau^2 &<\frac{n^2}{n^2 - \pi^2}e^{\frac{n^2}{n^2 - \pi^2}}\\
\tau &<\frac{n}{(n^2 - \pi^2)^{\frac{1}{2}}}e^{\frac{n^2}{2(n^2 - \pi^2)}}
\end{align*}
\end{proof}

\begin{lemma}
An Erd\"os-Renyi graph with constant $p$ has a unique Nash equilibrium for $\tau < \hat{\tau}=(\frac{1}{2\sqrt{np}-1})^{\frac{1}{2}}e^{\frac{1}{2(2\sqrt{np}-1)}}$.
\end{lemma}
\begin{proof}
For a Erd\"os-Renyi graph, with constant $p$ (\cite{Furedi1981})
\begin{align*}
\lambda_{\min}&=-c\sqrt{n}\\
\max{|\lambda_{\min}|}&=2\sigma \sqrt{n} + O(n^{\frac{1}{3}}\log{n}) & \text{where } \sigma=\sqrt{p}\\
\lambda_{\min}&>-2\sqrt{np}
\end{align*}
Substituting the value for $\lambda_{\min}$,
\begin{align*}
 W(\tau^2) &<\frac{1}{2\sqrt{np}-1}\\
\tau^2 &<\frac{1}{2\sqrt{np}-1}e^{\frac{1}{2\sqrt{np}-1}}
\end{align*}
\end{proof}

\begin{lemma}
A complete bipartite graph of size $n$ has a unique Nash equilibrium for $\tau < \hat{\tau}=(\frac{2}{n-2})^{\frac{1}{2}}e^{\frac{1}{n-2}}$.
\end{lemma}
\begin{proof}
\begin{align*}
\lambda_{\min}&=-\frac{n}{2}\\
W(\tau^2) &<\frac{2}{n-2}\\
\tau &<(\frac{2}{n-2})^{\frac{1}{2}}e^{\frac{1}{n-2}}
\end{align*}
\end{proof}
\fi

The quantity $\hat{\tau}$ of $G$ specifies the condition under which a unique Nash equilibrium exists. Table \ref{table:taus} details the value of $\hat{\tau}$ for various regular graphs (\cite{Ramach2015}).
\input{table_taus.tex}

Our observations on simple regular graphs give us an understanding of the behavior of the Nash Equilibrium in differnet types of settings. We see that for shorter lived information (content with smaller $\tau$), the process of sharing is relatively straightforward. In most graphs, for small $\tau < \hat{\tau}$, there exists a unique equilibrium. In symmetric graphs, this equilibrium is symmetric. In non-regular graphs, the equilibrium response is inversely related to the degree of a node since higher degree nodes can rely on good quality content through their many neighbors. Conversely, lower degree nodes tend to expend more effort since they have few neighbors that they can free ride on.

In general, more balanced graphs (with larger $\lambda_{\min}$) have less sensitivity to the ephemeral nature of information \ie the conditions for a unique equilibrium encompass a larger range of shelf life values. In more segregated graphs (with smaller $\lambda_{\min}$), the efforts of a few people can be enough for the graph as a whole and the equilibrium is less balanced in nature. 

Understanding the dependencies of the equilibrium in real world graphs is a little more challenging. Since these are not $d$-regular graphs, we do not expect symmetric equilibria to occur. In the case of the real world NYTimes graph, $\lambda_{\min}\approx -70$ (computed with python's sparse matrix package). Considering that the size of the NYTimes graph is $n=346k$ users, this case more closely resembles a balanced graph, like an Erd\"os-Renyi graph. For $\alpha=1$, a case where there is a relatively low cost of finding information, $\hat{\tau}\approx 0.12$ of the reference time period. For $\theta=1$hr (\ie., assuming readers' utility for content roughly compensate an effort to search every hour for new information), $\hat{\tau}\approx 7$min which is close to the empircally estimated shelf life of $\tau = 7.30$ min. 

\subsection{Tuning Shelf Life to Maximize Original Information} 
A media source would want to encourage users to spend more time on their site. Thus, they might be interested in tuning their parameter to maximize user effort. In a disconnected setting, each person is responsible for finding and consuming their own content. In this case, $y_{-i} = 0$ and the best response simplifies to $\phi(0) = \frac{\alpha}{\tau}W(\frac{\tau^\frac{\alpha+1}{\alpha}}{\alpha})$. At the value $\tau=\tau^*$, an individual is incentivitized to expend maximal effort.
\begin{claim} For an isolated node, $i$, the effort is maximized at $\tau^*=e^\frac{1}{\alpha+1}$.
\label{claim:isolated_maxima}
\end{claim}
\ifistr
\begin{proof}
The $\tau$ that corresponds to the maximum effort satisfies $\frac{\partial \phi}{\partial \tau}=0$. Further, since $i$ is isolated, $y_{-i}=0$. Hence, 
\begin{align*}
\frac{\partial \phi}{\partial \tau} = \frac{\partial \frac{1}{\tau}W(\frac{\tau^\frac{\alpha+1}{\alpha}}{\alpha})}{\partial \tau}&=0\\
\frac{\alpha}{\tau}\cdot \frac{1}{\alpha}\cdot \frac{\alpha+1}{\alpha}\tau^\frac{1}{\alpha}\cdot W'(\frac{\tau^\frac{\alpha+1}{\alpha}}{\alpha})+W(\frac{\tau^\frac{\alpha+1}{\alpha}}{\alpha})\cdot(-1)\frac{\alpha}{\tau^2} &=0\\
\text{Simplifying,} \quad W(\frac{\tau^\frac{\alpha+1}{\alpha}}{\alpha})&=\frac{1}{\alpha}\\
\tau &=e^\frac{1}{\alpha+1}
\end{align*}
It is easy to verify that this critical point is a maxima.
\end{proof}
\fi
In the case of symmetric graphs, there is always a symmetric equilibrium (Lemma \ref{lem:symm_ne}). We can calculate, for symmetric graphs, the $\tau^*$ that maximizes the amount of effort by any node in a symmetric equilibrium. 

\begin{claim} For an symmetric graph of degree $D$,  the effort in a symmetric equilibrium, $y_i$, is maximized at $\tau^*=\frac{e}{(1+D)^\alpha}^\frac{1}{\alpha+1} $
\label{claim:symm_maxima}
\end{claim}
\ifistr
\begin{proof}
Note that $y_{-i}=\frac{\alpha}{\tau(1+D)}W(\frac{\tau^\frac{\alpha+1}{\alpha}}{\alpha}(1+D))$ since $i$ has degree $D$ and the equilibrium is symmetric. 
Again, the $\tau$ that corresponds to the maximum effort satisfies $\frac{\partial \phi}{\partial \tau}=0$. By evaluating these expressions, we get $$W(\frac{\tau^\frac{\alpha+1}{\alpha}}{\alpha}(1+D))=1 \implies \tau=\frac{e}{(1+D)^\alpha}^\frac{1}{\alpha+1}$$
\end{proof}
\fi

\subsection{Specialization and Symmetry}
\label{sec:specialization}

We use simulations to examine how these theoretical results translate to various graph families, like complete graphs, star graphs, cycle graphs, complete bipartite graphs and Erd\"os-Renyi random graphs. For a graph family, we look at graphs of sizes ranging from $n=4$ to $n=400$ and edge density from $p=0.0001$ to $p=0.5$ (for Erd\"os-Renyi graphs). We then run an iterative algorithm that updates the best response until convergence \cite{Ramach2015} . The point of convergence (when it converges) is the Nash equilibrium. In the cases that we examined, the best responses converged to an equilibrium within 20 steps (though our algorithm does not guarantee convergence).

\ifistr
\begin{algorithm}[h]
\SetAlgoNoLine
\KwIn{Graph $G=(V,E)$, Shelf-life $\tau$, initialization value (optional)}
\KwOut{The amount of effort $\mathbf{y^*}={y^*_1, ..y_I^*, ..y_N^*}$}
$\mathbf{y}$ = 0;
\Repeat{$\delta > 0$}{
     \For{each node $i\in V$}{ 
            $y_{-i} = \sum_{j\in N(i)}{y_j}$\;
            bestResponse $ = \frac{1}{\tau}W(\tau^2 e^{-\tau y_{-i}})$\;
            $\delta = abs(y_i-$bestResponse$)$\;
            $y_i = $bestResponse\;
    }
 }
\caption{Finding the Nash Equilibrium by Iterative Update}
\label{alg:ne_iterative}
\end{algorithm}
\fi

Considering, first, the case of symmetric graphs (figure \ref{fig:reg_networks_diff_tau}), each line in the graph is the effort made by a particular node. Note that since many nodes have the same effort across different regimes of $\tau$, those lines overlapping each other and are hence not visible. In both the bipartite and cycle graph, in the specialized equilibrium, half the nodes overlap and expend most of the effort and the remaining half free-ride on those nodes. We see that, with shorter shelf-lives, individuals are more self-reliant. Conversely, longer shelf lives result in individuals relying on others efforts. 
Both cycle graphs and complete bipartite graphs exhibit the property that when content is long-term, the equilibria becomes more specialized with some individuals doing the majority of the work and others doing almost no work. Bipartite graphs split into their two partitions where those in one partition do all the work while those in the other do none. 

\begin{figure}[!ht]
\vspace{-5pt}
\subfloat{\includegraphics[width=0.245\textwidth]{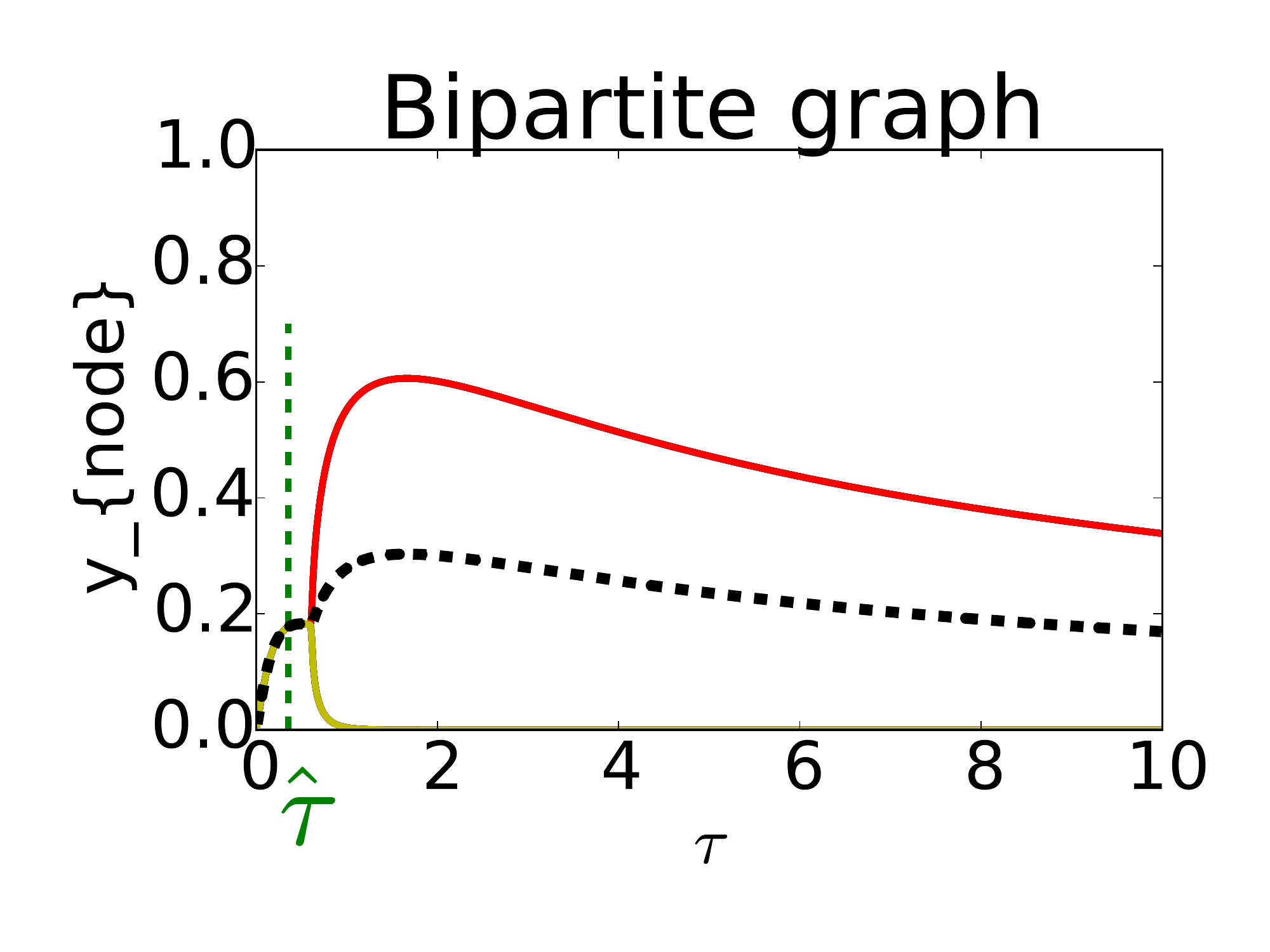}}
\subfloat{\includegraphics[width=0.245\textwidth]{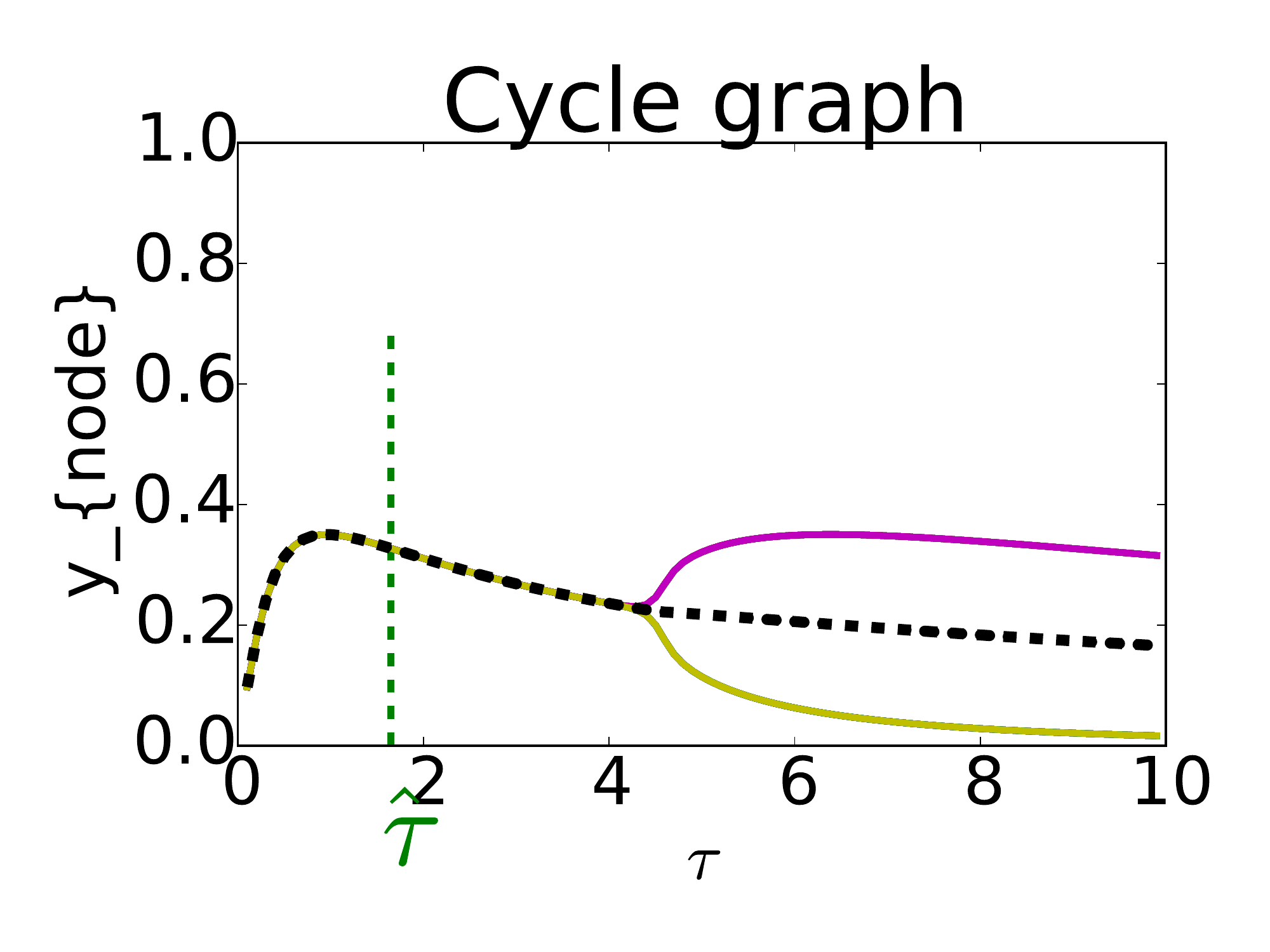}}
\vspace{-15pt}
\caption{Differing effort levels in the Nash Equilibrium (y-axis) with different $\tau$ (x-axis) in symmetric graphs. Each node (of $n=20$ nodes) is represented by a line in the figure. The unique equilibrium ($\tau<\hat{\tau}$) is always symmetric. (left) Complete bipartite graph (right) Cycle graph.}
\label{fig:reg_networks_diff_tau}
\end{figure}

The story is more complex in the case on assymetric graphs (figure \ref{fig:random_networks_diff_tau}). We consider the case of a star graph and an Erd\"os-Renyi graph, which gives us simple cases without the effect of heterogeneity. We also looked at a 10\% subset of a real world graph. 

\begin{figure}[!ht]
\begin{minipage}[c][5cm][t]{.2\textwidth}
  \vspace*{\fill}
  \centering
  \includegraphics[width=1\textwidth]{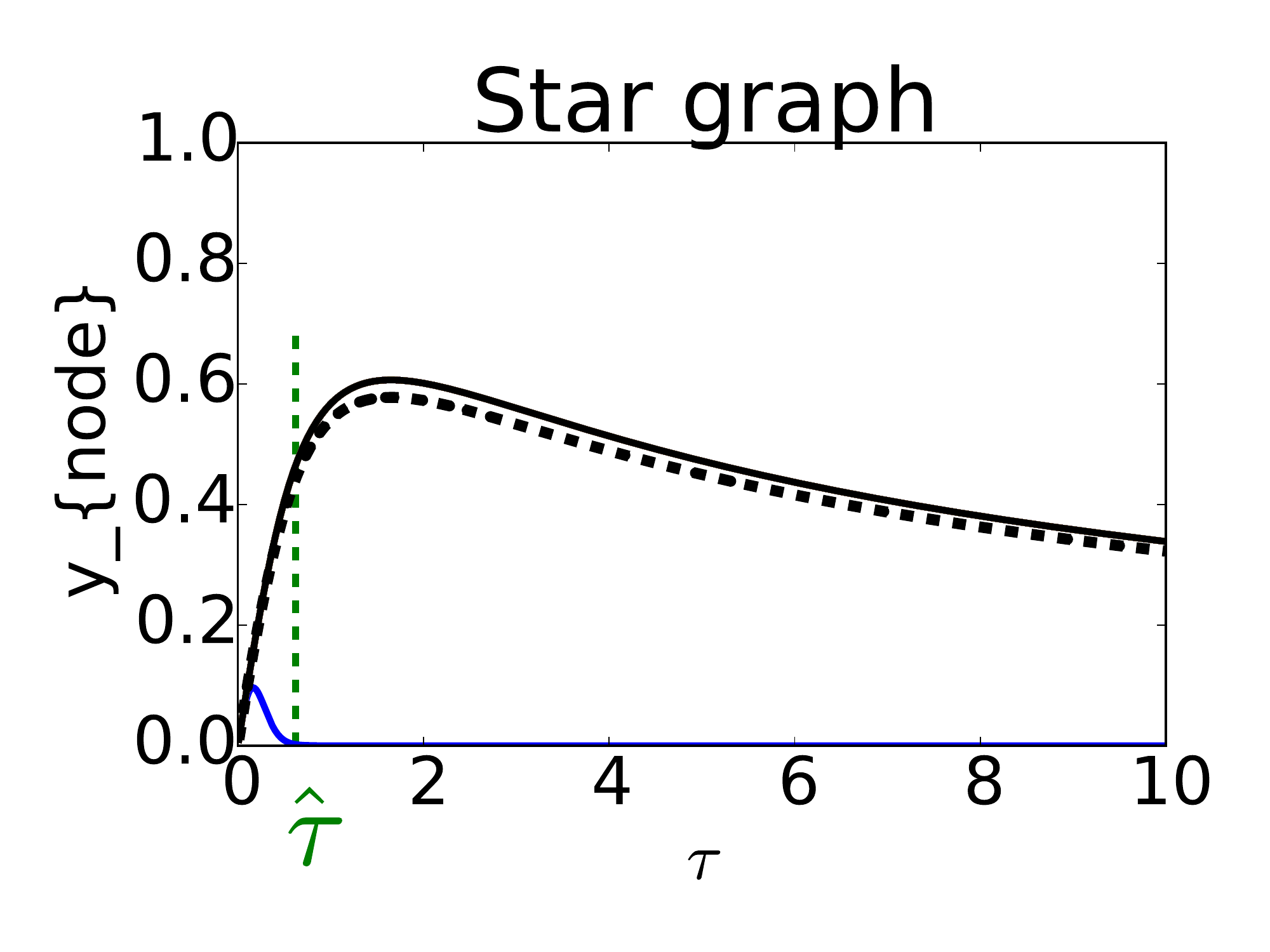}
\par\vfill
  \includegraphics[width=1\textwidth]{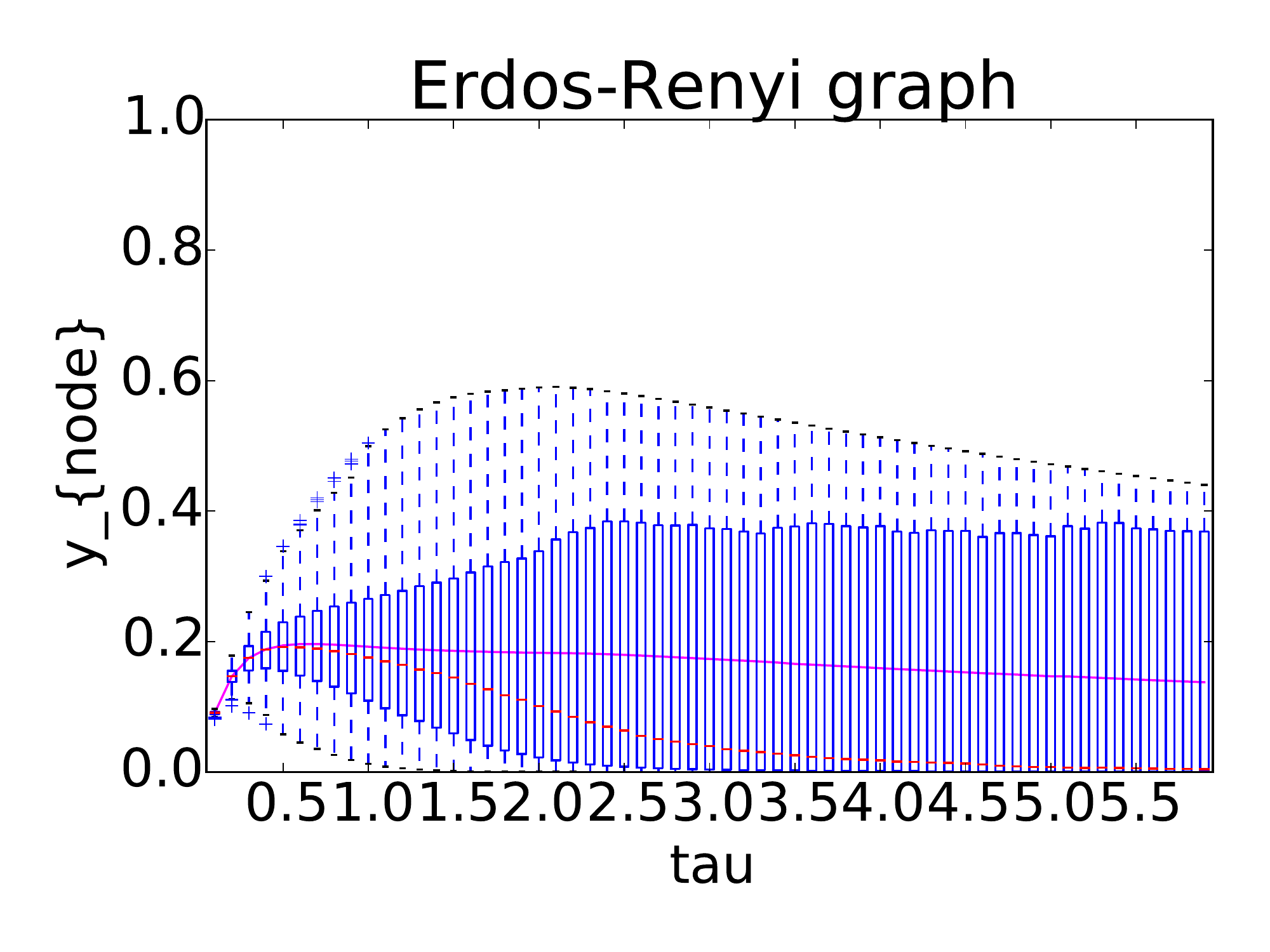}
\end{minipage}
\hspace{-10pt}
\begin{minipage}[c][5cm][t]{.295\textwidth}
  \vspace*{\fill}
  \centering
  \includegraphics[width=1\textwidth]{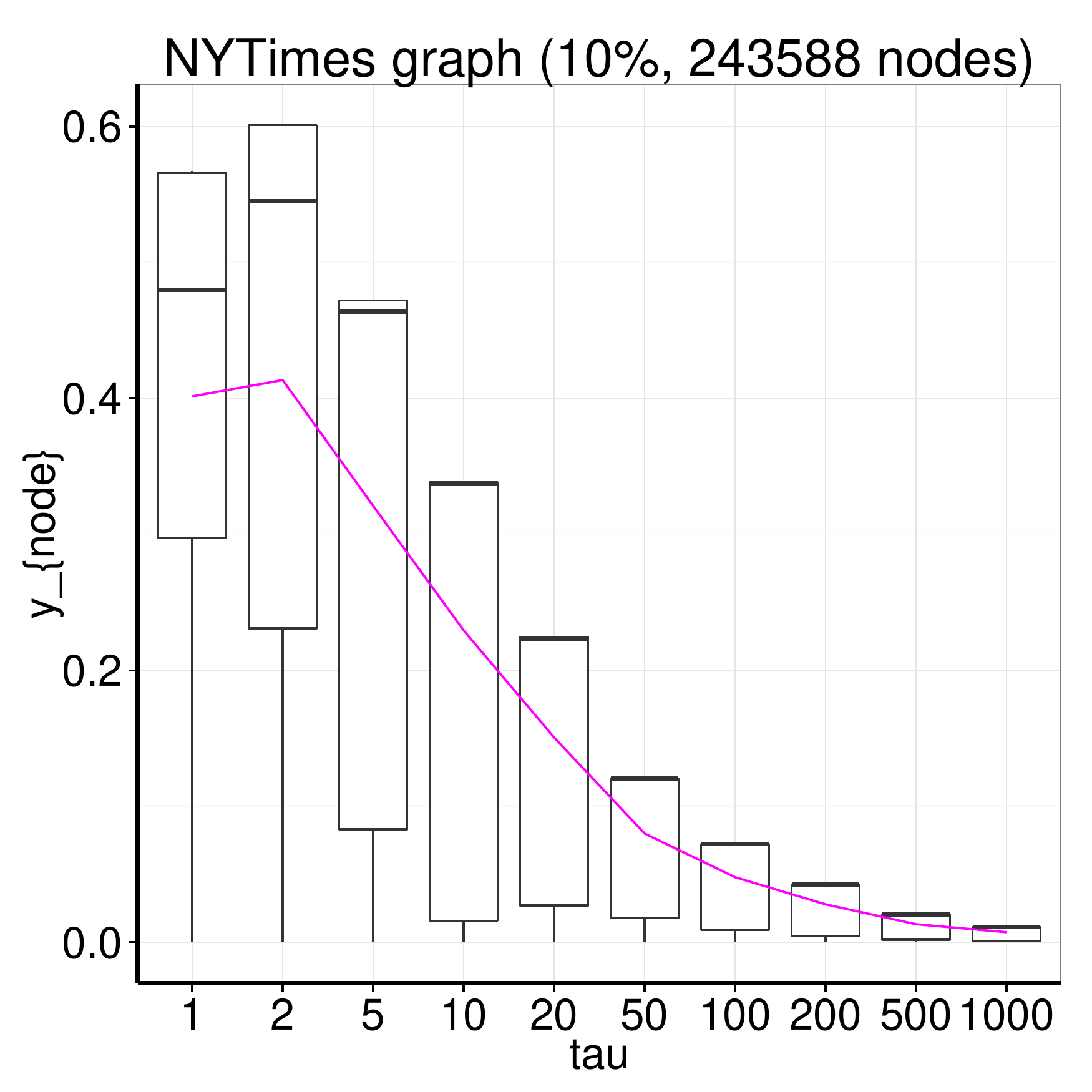}
\end{minipage}
\caption{Differing effort levels with different $\tau$ in assymetric graphs. Each bar represents the distribution of the amount of effort by all the nodes. The pink line is the average effort of all the nodes. (top-left) Star graph. (bottom-left) Erd\"os-Renyi graph ($n=1000,p = 0.01$). (right) Randomly sampled NYTimes graph with 243k nodes.}
  \label{fig:random_networks_diff_tau}
\end{figure}

In several of these graph families, we see that specialized equilibria occur. In the case of the star graph, the single central node does almost no work while all of his neighbors overlap and have much higher effort. In the case of random graphs or real world networks, it seems likely that a specialized equilibrium arises from the degree distribution of the nodes. However, in symmetric graphs, with all nodes having the same degree, clearly that is not the case. From lemma \ref{lem:symm_ne}, we know that a symmetric Nash equilibrium exists, but we observe that the system converges to a specialized Nash equilibrium. In the following section, we show that symmetric equilibria are not stable for large $\tau$.



\subsection{Theoretical Proof of Specialization}
When a unique Nash equilibrium exists, we understand the convergent network configuration. However, when there are multiple equilibria, it is not clear which of these configurations are realized --- for instance, some of these Nash equilibria can be unstable and, hence, never realized in practice. Here, we use the same definition of stability as in \cite{Bramoulle:2007va,Bramoulle:2010uj}. A Nash equilibrium is stable if a small change in the strategy of one player leads to a situation where two conditions hold: (i) the player who did not change has no better strategy in the new circumstance (ii) the player who did change is now playing with a strictly worse strategy.

Empirically, we observe that for longer-term content, the equilibrium for a cycle graph and a bipartite graph are specialized (figure \ref{fig:reg_networks_diff_tau}), \emph{inspite} of them being symmetric graphs. This indicates that the stability of the Nash equilibrium has some dependency on $\tau$. 
\begin{theorem} \ifistr [Specialization for Longer Shelf-Life] \fi
There exists an shelf-life $\tau$, such that, for any symmetric graph $G$ of degree $D\geq 3$, the symmetric equilibrium is not stable. 
\label{thm:stability}
\end{theorem}
\begin{proof}
The proof follows the outline of the Proof of Theorem 2 in \cite{Bramoulle:2007va}. 
It has two steps. The first step is a simple observation: 
$\text{If } \vec{y} < \vec{y'}\text{, then }\phi \circ \phi(\vec{y})<\phi \circ \phi(\vec{y}').$ 
This follows because the response function $\phi(y)$ is a decreasing function of $y$. 

The second step is to show that under some small perturbation $\vec{\epsilon}>0$, we have 
$\phi\circ \phi(\vec{y}+\vec{\epsilon})> \vec{y}+\vec{\epsilon}$
(here the vector inequality $\vec{x}>0$ corresponds to coordinate wise inequality $x_i >0 ~\forall i$). In other words, with any small change from the equilibrium, the best response moves further away (strictly) from the equilibrium. This shows that the equilibrium is not stable in the sense of \cite{Bramoulle:2007va,Bramoulle:2010uj}. For simplicity's sake, we consider only a quadratic cost function.
  
Let $\mathbf{\tilde{y}}$ be the symmetric equilibrium in the symmetric graph of degree $D$. Then, $\tilde{y_i}=\tilde{y}, \forall i$. Note that $\tilde{y}=\phi(\tilde{y})$ because it is an equilibrium. Here, we perturb all the responses by some $\epsilon>0$
\begin{align*}
\phi(\vec{y}+\vec{\epsilon})&=\phi(\vec{y})+\nabla \phi \cdot \vec{\epsilon}\\
\phi_i(\vec{y}+\vec{\epsilon})&=\phi_i(\vec{y})+D\frac{\partial \phi_i}{\partial y_j}\epsilon \quad \text{for some } j\in N(i) 
\end{align*}
since $\frac{\partial \phi_i}{\partial y_j} = 0 \quad \text{if } j\notin N(i)$ and equal otherwise. Similarly, 
\begin{align*}  
\phi_i\circ \phi(\vec{y}+\vec{\epsilon}) &=\phi_i([\dots,\phi_j(\vec{y}+\vec{\epsilon}),\dots])\\
&=y_i+D^2\big(\frac{\partial \phi_i}{\partial y_j}\big)^2\epsilon \qquad \text{any } j\in N(i)
\end{align*}
To show that the symmetric equilibrium is not stable, we need
\begin{align*}
y_i + D^2\big(\frac{-W(\tau^2e^{-\tau \tilde{y}})}{1+W(\tau^2e^{-\tau \tilde{y}})}\big)^2\epsilon&> y_i+\epsilon\\
W(\tau^2e^{-\tau \tilde{y}})&> \frac{1}{D-1}
\end{align*}
In other words, we want $\tau^2e^{-\tau \tilde{y}}> \frac{1}{D-1}e^{\frac{1}{D-1}}$. Substituting for $\tilde{y}$ (lemma \ref{lem:symm_ne}) and simplifying, we get that the symmetric Nash equilibrium is not stable when 
$$2\ln{\tau}-\frac{1}{(D+1)}W(\tau^2(D+1))  > -\ln{(D-1)}+\frac{1}{D-1}.$$

Setting $\tau$ to be a constant (\eg $\tau=10$), one only needs to verify that the following holds: $W(D+1)<(D+1)(\ln{(D-1)}+2\ln{\tau})-\frac{D+1}{D-1}$, which is true for $D>3$.
\end{proof}

%% file: table_taus.tex
\begin{table}[h]
\caption{Conditions for unique Nash Equilibrium ($\tau < \hat{\tau}$) for graphs with $n$ nodes ($\alpha=1$)
\label{table:taus}}{%
\begin{tabular}{| l | c | c | c |}
\toprule
Graph     & $\lambda_{min}$ & $\hat{\tau}$\\
\midrule
Complete & $-1$ &$\forall \tau$ ($\infty$) \\
Cycle (Even)  & $-2$ & $\sqrt{e}$ \\
Cycle (Odd)   & $-2 + \frac{\pi^2}{n^2}$ & $\frac{n}{(n^2 - \pi^2)^{\frac{1}{2}}}e^{\frac{n^2}{2(n^2 - \pi^2)}}$ \\
Erd\"os-Renyi & $-2\sqrt{np}$ &$(\frac{1}{2\sqrt{np}-1})^{\frac{1}{2}}e^{\frac{1}{2(2\sqrt{np}-1)}}$\\
Star & $-\sqrt{n-1}$ & $(\frac{1}{\sqrt{n-1}-1})^{\frac{1}{2}}e^{\frac{1}{2(\sqrt{n-1}-1)}}$\\
Complete Bipartite& $-\frac{n}{2}$& $(\frac{2}{n-2})^{\frac{1}{2}}e^{\frac{1}{n-2}}$ \\
\bottomrule
\end{tabular}}
\end{table}

%% file: relwor.tex
Our contributions relate and contribute to several directions of research: 

\textbf{(1) Studies of online diffusion of information} have previously established the importance of content produced by mass media in online diffusion. They highlight in particular that news typically reaches a large audience not directly but through a set of influencers or connectors~\cite{Cha:2012kz, Wu:2011cd}. This result confirms the classical hypothesis of a two-step information flow~\cite{Katz:1957js}, and was shown to have additional benefits, such as broadening the range of opinions seen by a user~\cite{An:2011ur}. However, the dynamics of participation and influence remains elusive. For instance, relying on number of followers to judge an influencer can be misleading~\cite{Cha:2010tr,Bakshy:2011uw} and predicting who is successful at an individual level was shown to be generally unreliable~\cite{Bakshy:2011uw}. 
Our work takes a different starting point: We follow evidence that a large fraction of diffusion cascades occur close to a seed node~\cite{Goel:2012io}. Hence we focus on identifying those who contribute in adding \emph{original} content in the network, and how this relates to temporal characteristics of the content being exchanged. Previous studies of temporal properties of diffusion typically focused on leveraging that those are short-lived~\cite{Cha:2009vj,Rodriguez:2011vy}, or on using patterns in the time series for better classification~\cite{Kwon:2014vh,Kamath:2013tf,Yang:2011wt}. 

\textbf{(2) Analysis of the private provision of public goods}, or investments made by players that more generally affect the outcome of others, originally emerged to inform public policy. Its most celebrated result, the \emph{neutrality principle}~\cite{Bergstrom:1986ej}, states that the investment produced by a group is entirely carried by most wealthy individuals, and is insensitive to income redistribution. This, however, holds only for a \emph{global} public good in which all players are equally affected by others, and recently was shown not to generalize beyond regular graphs~\cite{Allouch:2012do}. The general network case was studied more recently~\cite{Ballester:2006hs,Bramoulle:2007va,Bramoulle:2010uj}, typically in a model assuming that a player’s best response follows from other player’s actions in a linear matrix form. Even for that simple case, predictions vastly differ: On the one hand, a study of small effects~\cite{Bramoulle:2010uj} proves that the system converges to a unique equilibrium in which all participate. 
On the other hand, more general cases prove that specialization is unavoidable, and that multiple equilibria can be attained~\cite{Bramoulle:2007va}. Our analysis extends those results by providing the first non-linear dynamics for which a similar dichotomy can be proved; in particular, it proves that a simple model of perishable public goods leads to either of these behaviors depending on the product lifespan. 

\textbf{(3) The role of elites in information acquisition} has been studied in very different contexts such as social learning~\cite{Bala:1998cr,Acemoglu:2008uy} and opinion formation~\cite{Golub:2010er,Acemoglu:2009ui}. 
Those results are different in spirit as they typically focus on aggregation of multiple contributions on the same specific topic, either within a social networks or in the presence of a kernel of experts. For that reason, they typically assume specific types of information or interactions. Our model focuses on a simpler model in which information can be produced under some exerted effort, but is free to reproduce within a given network. The work motivated similarly to ours considers a similar process in an endogenous network where players may create new links at a fixed cost~\cite{Galeotti:2010tj}. It was shown that these dynamics typically lead to extreme specialization, even among \emph{ex ante} identical players. However, Heterogeneous systems can't be analyzed in the same manner, and networks produced are typically very schematic (bi-partite). Our work proves that specialization emerges in an exogenous network, even without the reinforcing process of strategic link formation. 

%% file: conclusion.tex
Knowledge sharing has been greatly facilitated by social network services. Increasingly, it affects businesses, political debates and public services. Yet, after years of measurements, the structure of online diffusion remains complex and was shown to vary across media and topics. Our results identify, for the first time, how the shelf life of information affects its diffusion. This leads to various types of specialization that can all be described in the unifying theory of public good. 

While we empirically observe a remarkable match to the theoretical predictions on a qualitative level, we would like to point out that the current model of public good we introduce is highly idealized, especially as it assumes homogeneous cost of information acquisition. Proving that specialization occurs even in such symmetric cases is, in a sense, a worst-case result. In reality, several other factors contribute to users exerting higher effort in information acquisition including enjoyment~\cite{Feick:1987ho}, which typically varies across users depending on topics. However, our results generalize to heterogeneous perishable public goods to predict, for instance, that a single equilibrium exists whenever shelf life is sufficiently small. The qualitative effect of shelf life should also remain since our empirical observations prove it, even in a large number of very different mass media sources. We do, however, observe some amount of variance within this trend and accounting for other previously identified factors to predict span of content diffusion more accurately seems a promising direction.

Whenever public good theory allows for simple equilibrium computation, i.e. for short lived content, it also yields additional insight on how to locally or globally optimize content to encourage more participation. Ultimately, testing if those insights provide algorithms to design effective incentives to users for enhanced participation offers a way to validate those claims.